\documentclass[a4paper,10pt]{article}
\pdfoutput=1
\usepackage{fullpage}
\usepackage{graphicx}
\usepackage{amssymb,amsmath}
\usepackage{xfrac}
\usepackage{amsthm}
\usepackage{multirow}
\usepackage{color}
\usepackage{authblk}
\usepackage{listings}
\usepackage{algorithm}
\usepackage{algorithmic}
\usepackage[font={footnotesize}]{caption}
\usepackage[font={footnotesize}]{subcaption}
\usepackage{hyperref}
\usepackage[colorinlistoftodos,prependcaption,textsize=normal]{todonotes}
\usepackage{changepage}
\usepackage{lipsum} 
\usepackage{comment}

\usepackage{hyphenat}
\newtheorem{lemma}{Lemma}

\newcommand{\rev}[1]{{\color{black} #1}}

\begin{document}

\title{Maximising the Utility of Enterprise \\Millimetre-Wave Networks}

\author[1]{N. Facchi} 
\author[2]{F. Gringoli}
\author[3]{P. Patras} 
\affil[1]{Deptartment of Information Engineering and Computer Science, \mbox{University of Trento}}
\affil[2]{Deptartment of Information Engineering, CNIT / University of Brescia}
\affil[3]{School of Informatics, University of Edinburgh}

\date{}

\maketitle

\begin{abstract}
Millimetre-wave (mmWave) technology is a promising candidate for meeting the intensifying demand for ultra fast wireless connectivity, especially in high-end enterprise networks. Very narrow beam forming is mandatory to mitigate the severe attenuation specific to the extremely high frequency (EHF) bands exploited. Simultaneously, this greatly reduces interference, but generates problematic communication blockages. As a consequence, client association control and scheduling in scenarios with densely deployed mmWave access points become particularly challenging, while policies designed for traditional wireless networks remain inappropriate. In this paper we formulate and solve these tasks as utility maximisation problems under different traffic regimes, for the first time in the mmWave context. We specify a set of low-complexity algorithms that capture distinctive terminal deafness and user demand constraints, while providing near-optimal client associations and airtime allocations, despite the problems' inherent NP-completeness. To evaluate our solutions, we develop an \mbox{NS-3} implementation of the IEEE 802.11ad protocol, which we construct upon preliminary 60GHz channel measurements. Simulation results demonstrate that our schemes provide up to 60\% higher throughput as compared to the commonly used signal strength based association policy for mmWave networks, and outperform recently proposed load-balancing oriented solutions, as we accommodate the demand of 33\% more clients in both static and mobile scenarios. 
\end{abstract}

\section{Introduction}

%Users increasingly prefer wireless connectivity whilst high performance applications, including ultra high definition (4K) video, wire-equivalent docking, virtual reality streaming, and low latency data upload/download, are proliferating \cite{intel:2016}. 
Users' predilection for wireless connectivity is increasingly incompatible with the stringent performance requirements of emerging applications, including uncompressed ultra high definition (HD) video, wire-equivalent docking, virtual reality streaming, and low latency data upload/download \cite{intel:2016}.
In response, the industry is exploring the use of license exempt extremely high frequencies (mmWave) in the 60GHz band, for short range multi-gigabit per second wireless communications \cite{Boccardi:2014}. These efforts have already materialised as new standard amendments, e.g. IEEE 802.11ad~\cite{802.11ad}, %and 802.15.3c~\cite{802.15.3c},
recently unveiled WiGig routers~\cite{arstechnica:2016}, 
and business-oriented laptops%\footnote{Arstechnica, ``TP-Link unveils worlds first 802.11ad WiGig router'' (2016), {\ttfamily https://tinyurl.com/tp-link-11ad}} 
~\cite{engadget:2016}.

Different to legacy wireless solutions, mmWave technology leverages vast spectral resources (up to 2GHz-wide channels) currently underutilised. Their potential, however, can only be realised through highly directional digital beamforming, since signals attenuate dramatically in this frequency range \cite{Tie:2012}. %and do not penetrate obstacles \cite{Tie:2012}. 
Forming narrow beams not only mitigates fading, but also reduces interference between adjoining TX/RX pairs. Consequently, links between stations and access points (APs) can be regarded as pseudo-wired and channel access no longer subject to collisions. The caveat is that \emph{associated clients are shut out whenever an AP communicates with anyone of their neighbours}. To ensure all stations are given opportunities to receive and/or transmit packets, the IEEE 802.11ad standard defines a Service Period (SP) based channel access mechanism, though the task of scheduling SPs is deliberately left open to accommodate proprietary implementations~\cite{802.11ad}. 

This problem is further complicated in enterprise wireless networks, including stocks trading offices, broadcasting studios that manipulate raw ultra HD video,\footnote{See BBC IP Studio, http://www.bbc.co.uk/rd/projects/ip-studio} and emerging tactile Internet environments.\footnote{E.g. http://www.huawei.com/minisite/5g/en/touch-internet-5G.html}
There, mmWave clients will often lie within the range of multiple APs that serve different numbers of stations, as exemplified in Fig.~\ref{fig:simplescenario}, possibly having dissimilar traffic demands. In this scenarios, \emph{the challenge is deciding both to which AP to associate clients and what airtime budget to allocate for each.} With an appropriate logic that is yet to be developed, such decisions could be enforced by \emph{central controllers} similar to those widely used in today's enterprise wireless networks to load balance clients over the available APs and bands.~\footnote{This is the case of commercially available solutions including Cisco WLC (\url{http://www.cisco.com/c/en/us/products/wireless/wireless-lan-controller/index.html}) and Aruba Mobility Controller (\url{http://www.arubanetworks.com/products/networking/controllers})}

\begin{figure}[t]
	\centering
	\includegraphics[width=0.6\columnwidth]{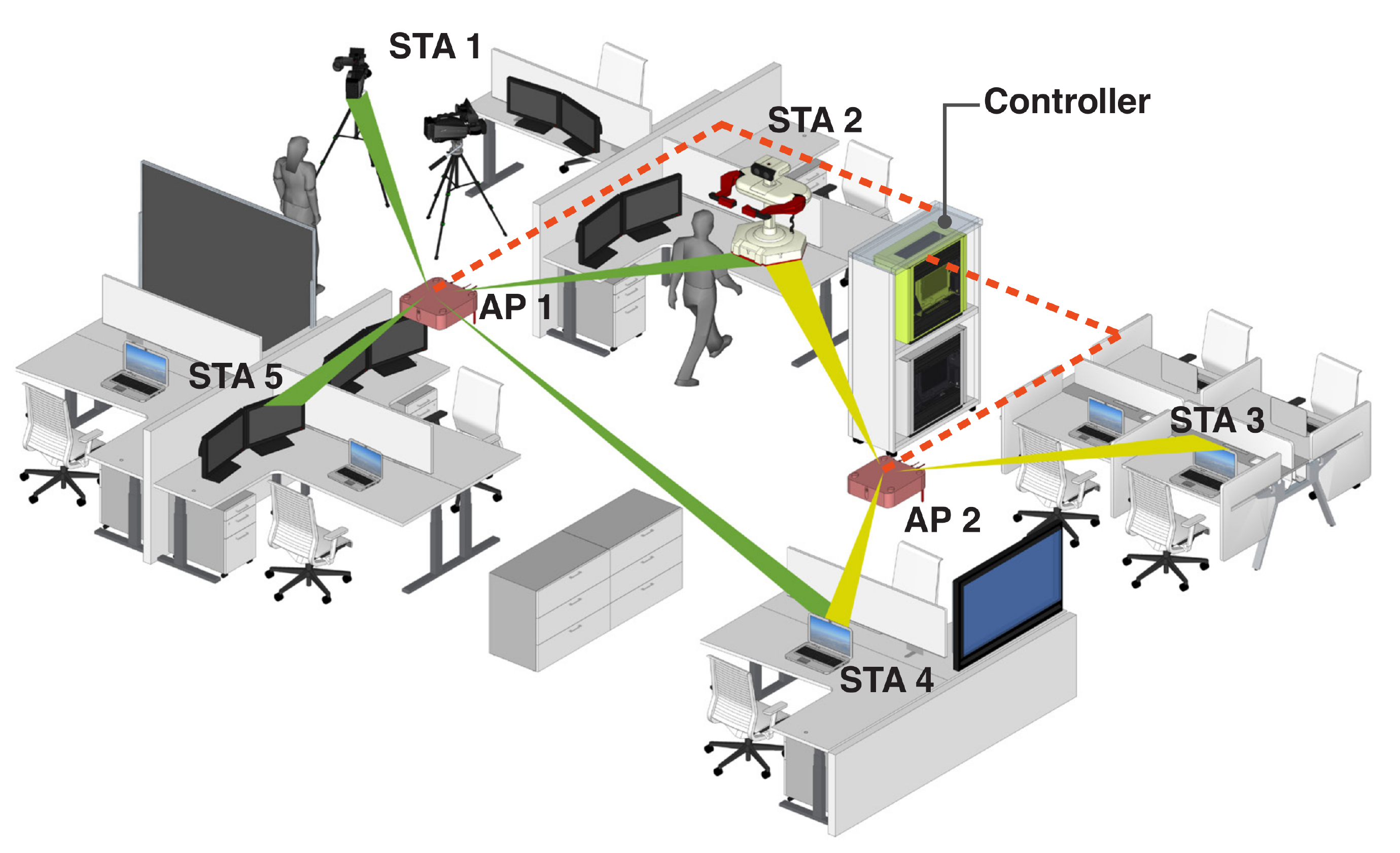}
	\caption{Simple example of the envisioned system, comprising two APs and five active stations. APs are connected to the central controller that runs the algorithms we introduce in this paper to find optimal client associations and airtime allocations that maximise network utility. Clients and APs communicate over directional links (shaded beams). Stations 2 and 4 are within the range of both~APs.}
	\label{fig:simplescenario}
%\vspace*{-1em}
\end{figure}

Commonly adopted signal strength based association policies are oblivious to load conditions \cite{Giannoulis:2013} and thus may lead to inappropriate decisions in mmWave networks. Likewise, client association mechanisms for traditional 802.11 wireless networks~\cite{Gong:2012,Athanasiou:2009} or cellular systems (e.g. \cite{Ye:2013}) are ill suited to mmWave, due to the substantial differences between these technologies. Association control and SP allocation in mmWave networks are largely unexplored; recent solutions focus primarily on load balancing, downplaying airtime budget constraints, and requiring non standard signalling~\cite{Athanasiou:2015}. Without carefully controlling which AP serves each client and for how long on average, we argue that the overall network throughput performance will be sub-optimal and user demand often unsatisfied, even when sufficient resources are available in the network.

\textbf{Contributions:} In this paper we formulate and solve the client association control and SP allocation tasks in high-end mmWave networks as utility maximisation problems, \emph{capturing the severe terminal issues unique to such systems}. We consider general scenarios with both backlogged stations and clients with finite load requirements, which encounter heterogeneous link qualities to the APs within range, and may be either static or mobile. We use the same definition of utility as given by F. Kelly, i.e. the sum of the logarithms of individual station throughputs~\cite{Kelly:1997}, which strikes a \emph{good trade-off between maximising network throughput and providing airtime fairness}. We envision a centralised network driven architecture (as in Fig.~\ref{fig:simplescenario}) that could be built upon recent advances in software-defined networking (SDN) \cite{sdn}, and IEEE 802.11 protocol amendments for wireless network~\cite{802.11v} and radio resource management~\cite{802.11k}. These would enable the central controller to collect information from the deployed APs and their clients, and \emph{enforce the computed client--AP associations and airtime allocations in a standard compliant fashion}. 
Specifically,
\begin{enumerate}
 \item We show that under backlog conditions, the utility optimisation problem we pose is NP-complete, but its relaxed version is convex. We use well-established Lagrangian tools to solve the relaxed version and give a linear complexity iterative rounding algorithm, to derive solutions to the original problem;
 \item For finite load scenarios we introduce an algorithm that captures terminal deafness and traffic constraints, and combines simulated annealing with airtime water filling techniques, to find near-optimal association matrices and airtime allocation vectors almost in real-time; 
 \item Using an NS-3 based 802.11ad simulation module we develop, building on preliminary 60GHz channel measurements, we demonstrate that our solutions provide up to 60\% higher total throughput as compared to the standard's default signal strength based policy, while satisfying the demand of 33\% more clients, in comparison with the recently proposed DAA mechanism~\cite{Athanasiou:2015}.%, as we are able to accommodate the traffic demand of all stations under finite load scenarios.
\end{enumerate}

To the best our knowledge, this is the first attempt to cast client association control and airtime allocation as utility maximisation problems in the mmWave context, whilst the solutions we provide are demonstrably effective under a broad range of network conditions.

%sThe rest of the paper is organised as follows. We outline the system architecture in Sec.~\ref{sec:architecture}, analyse individual throughput in Sec.~\ref{sec:analysis}, formulate utility maximisation problems and give solutions for backlogged and finite load scenarios in Sec.~\ref{sec:satcond} and~\ref{sec:finiteload}. We evaluate the performance of our schemes in Sec.~\ref{sec:evaluation}, give guidelines for deploying them in practical networks in Sec.~\ref{sec:implementation}, and discuss related work in Sec.~\ref{sec:related}. We conclude and discuss future directions in Sec.~\ref{sec:conclusion}.  \paul{(We may leave this paragraph out.)}

\section{System Architecture}
\label{sec:architecture}

We consider enterprise mmWave wireless network deployments with $M$ access points and $N$ clients. To thwart high signal attenuation inherent to the 60GHz band, each station (AP or client) is equipped with a transceiver that can digitally form and steer beams of very narrow widths, to transmit or to receive packets. Therefore, interference levels can be considered negligible and not impacting on the achievable bit rates. This is in line with recent studies \cite{Singh:2011} that confirm even uncoordinated packet exchanges between different transmit--receive pairs experience very small collision probabilities. \rev{While we acknowledge that interference increases in outdoor deployments as cell density grows~\cite{rebato2016understanding}, this is not applicable to indoor scenarios where ceiling-mounted access points will experience no interference when the angular separation between links is as little as 10-12$^\circ$~\cite{Zhu:2014b}. This is feasible in office environments even with consumer-grade equipment whose antenna patterns exhibit side lobes~\cite{Nitsche:2015}.} As such, in our setting mmWave links can be regarded as \emph{pseudo-wired} point-to-point connections.

\begin{figure}[t]
	\centering
	\includegraphics[width=0.8\columnwidth]{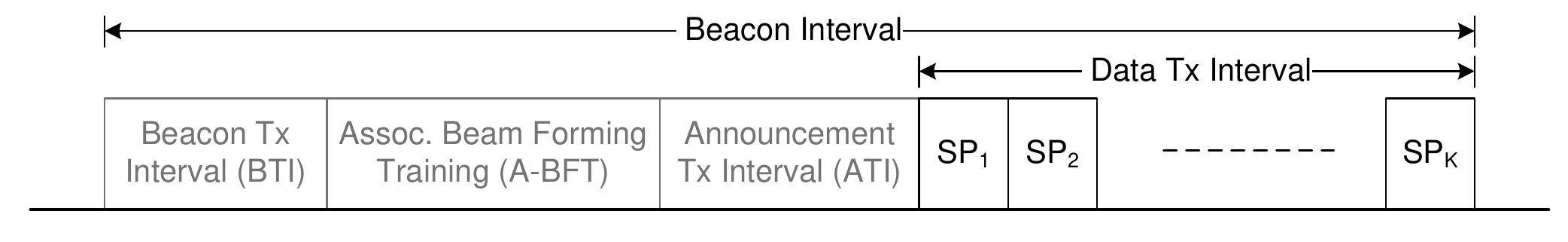}
	\caption{IEEE 802.11ad super-frame. Frame transmissions are performed in a scheduled fashion during the Data Tx Interval using service periods (SPs) \cite{802.11ad}.}
	\label{fig:superframe}
	%\vspace*{-1em}
\end{figure}

We focus on emerging enterprise networks with SDN capabilities that enable controllers to manipulate the configuration of APs (and their clients) via protocols such as NETCONF~\cite{rfc6241}.
In addition, APs can request neighbour and link measurement reports from stations, through 802.11k primitives~\cite{802.11k}, which a centralised controller will use to enforce control association to specific APs, as determined by the algorithms we propose in this work. \rev{mmWave channel sounding capable of measuring multi-path delays with 2ns granularity has been recently demonstrated~\cite{MacCartney:2017} and APs can represent per-client PHY rate measurements as 2-byte \emph{(client\_ID, PHY\_rate)} tuples whose transmission every beacon interval incurs negligible overhead.} 
Network assisted association requests as we proposed will be underpinned by 802.11v enhancements \cite{802.11v}. %, which end-user device manufacturers have begun implementing~\cite{Apple:2016}. 
Upon bootstrap, we assume clients associate following the de facto signal-to-noise ratio (SNR) based procedure, as stipulated by the standard~\cite{802.11ad}. 

Given the carrier-grade requirements specific to business-oriented networks, we
focus on 802.11ad mmWave networks where APs operate in a
scheduled fashion using the Service Period (SP) channel access scheme. Channel time is divided into fixed-size beacon intervals ($\sim$100ms)
that carry at the beginning synchronisation information, beam form training
related signalling, and announcement of the subsequent SPs assigned
\cite{802.11ad}. The duration of these is variable, while stations may aggregate
multiple data frames within their SPs, to improve protocol efficiency. We
summarise this behaviour in Fig.~\ref{fig:superframe}. By this approach, clients associated to the same AP do not contend for channel access and each AP serves only one associated client at a time during non-overlapping scheduled SPs, whilst multiple APs could simultaneously serve different clients.
We confine consideration to protocol features already specified by the approved 802.11ad standard~\cite{802.11ad}. We note however that the association control and airtime allocation solutions we propose herein could be easily adapted to encompass more aggressive modulation schemes, channel bonding, and MU-MIMO enhancements, which are candidates for future standard amendments, e.g. 802.11ay~\cite{802.11ay}.

The potential of mmWave networks can only be realised using highly directional narrow beams and any two communicating stations must know how to configure
their antennas for beam alignment. This is achieved through antenna and beamforming training, and the standard specifies a set of protocols for this purpose, but does not provide a precise rule on how to employ these. Importantly the standard does not specify how to select the best antenna configuration, the training process being complex
and chip implementation dependant.\footnote{Mechanisms to enable learning with high accuracy the relative
positions of devices in the network, currently under discussion in the IEEE 802.11az task group~\cite{802.11az}, may alleviate the complexity of beamforming.} 
For these reasons, in this work, we assume all stations are capable of performing reliable beamforming training, while we do not consider the precise details of the
actual mechanism. However, our modelling and the NS-3 based simulator we develop take into account the possible overhead due to beamforming protocols. 
\rev{Our focus is on the optimal client association and airtime allocation tasks, which have been largely overlooked by the research community. mmWave channel modelling remains outside the scope of our work.}

\section{Throughput Analysis}
\label{sec:analysis}

In this section we formalise the throughput performance of mmWave clients in the
envisioned enterprise scenario. Recall that we aim to meet strict
quality-of-service requirements and therefore medium access control is regulated
by the APs using the SP mechanism. We consider realistic multi-rate conditions,
i.e. clients employ Single Carrier or OFDM PHY modulation and coding schemes
that depend on the signal-to-noise ratio (SNR) on the links to the APs and the
channel bandwidth. Note that the SNR may be subject to link blockage due to human movement or other obstructions
on the direct propagation path between a client and an AP. Our model captures such circumstances as low achievable 
rates, or null bit rates when establishing a communications link may not be (temporarily) feasible. 
We denote by $r_{i,j}$ the bit rate a client $i$ achieves
when transmitting to an AP $j$ during an SP.

Consider a binary association vector $\mathbf{x}$, where an element $x_{i,j} \in
\{0,1\}$ indicates whether a client $i$ is associated to an AP $j$
($i=1,\ldots,N,$ and $j=1,\ldots,M$), i.e.

\begin{equation}
  \small
	\label{eq:decvar}
	x_{i,j} = \begin{cases}
		1, & \quad \text{if client $i$ is associated to AP $j$}; \\
		0, & \quad \text{otherwise}.
	\end{cases}
\end{equation}

We initially assume that all stations are backlogged (saturation conditions) and
later relax this assumption to account for general finite load scenarios.  Under
these circumstances, maximising network utility has been shown to correspond to
allocating equal airtimes to all stations connected to a given AP, irrespective
of their bit rates \cite{Patras:2015}. As such, our goal is to allocate a
$t_{i,j}$ fraction of an AP $j$'s total airtime budget to a connected client
$i$. Formally,

\begin{equation}
\small
\label{eq:satairtime}
 t_{i,j} = \frac{T_j - O_j}{\sum_{k=1}^{N} x_{k,j}}, \forall i,j,
\end{equation}

\noindent where $T_j$ denotes the duration of a super-frame (beacon interval) as enforced
by AP $j$ and $O_j$ is the protocol overhead due to beacon transmission,
(optional) beamforming (BF) training, and management operations (see
Fig.~\ref{fig:superframe}). Given the scheduled nature of the medium access in
802.11ad, hereafter we will use the terms airtime and service period (SP)
interchangeably.

With the above, the throughput $S_{i,j}$ obtained by client $i$ when connected
to AP $j$ is given by

\begin{equation}
  \small
	\label{eq:achtp}
	S_{i,j} = \frac{r_{i,j}t_{i,j}}{T_j} = \frac{h_j r_{i,j}}{\sum_{k=1}^{N} x_{k,j}},
\end{equation}

\noindent where $h_j =  (T_j-O_j)/T_j$. In what follows we only consider feasible
associations, i.e. those for which a client $i$ falls in the coverage of an AP
$j$ and thus $r_{i,j} \neq 0$.

\section{Utility Maximisation for Saturation Scenario}
\label{sec:satcond}

Our objective is to find the client \emph{association matrix} $\mathbf{x}$ that
\emph{maximises the total utility of the network}, i.e. solve the following
optimisation problem:

\begin{align}
  \small
	&\max_{\mathbf{x}} U := \sum_{j=1}^{M} \sum_{i=1}^{N} x_{i,j}\log{S_{i,j}}, \label{eq:goal} \\
	\text{s.t.}\
%	&S_{i,j} \leq \frac{h_j r_{i,j}}{\sum_{k=1}^{N} x_{k,j}}, \forall i,j; \quad ~~\text{(throughput feasibility)}\\
	&\sum_{j=1}^{M} x_{i,j} = 1, \forall i; \qquad \qquad \quad \text{(single AP association)} \label{eq:constroneap} \\
  	&x_{i,j} \in \{0,1\}, \forall i,j.      \qquad \qquad \qquad ~\text{(function domain)}\label{eq:constrdecvar}
\end{align}

\begin{lemma}
 The optimisation problem specified by (\ref{eq:goal})--(\ref{eq:constrdecvar})
	is NP-complete.
\end{lemma}

\begin{proof}
Denote $P$ the problem given in (\ref{eq:goal})--(\ref{eq:constrdecvar}). By (\ref{eq:achtp}), $S_{i,j}$ is a function of the inverse sum of some terms $x_{k,j}$, for all $j$. Therefore the objective of the problem posed is a non-linear function of variables in the $\{0,1\}^N$ set (\ref{eq:constrdecvar}). Now, consider a simpler problem $P'$ where a client $i$ attains a small constant throughput $\theta_j$ when connected to AP $j$, irrespective of the number of clients this servers. The objective (\ref{eq:goal}) becomes $\sum_{j=1}^{M} \sum_{i=1}^{N} \theta_j x_{i,j}$ and since $x_{i,j} \in \{0,1\}$, the constraint (\ref{eq:constroneap}) is equivalent to $\sum_{j=1}^{M} x_{i,j} \leq 1$. It follows that $P'$ is an instance of the \mbox{0--1} knapsack problem, which by Theorem 15.8 in \cite{Papadimitriou:2000} is NP-complete. Since a solution to $P$ can be verified, $P$ is NP, and as $P > P'$, while $P'$ is NP-complete, then $P$ is NP-complete.

\end{proof}

Finding a solution to this type of optimisation problems within reasonable time
is known to be difficult \cite{Arora:2009}. Consequently, we proceed with a
relaxation of our original problem, replacing the constraint $x_{i,j} \in
\{0,1\}$ and allowing $x_{i,j}$ in the $[0,1]$ interval (Fractional User
Association). This is similar in nature to the linear programming relaxation of
the set cover problem studied by Lov\'asz~\cite{Lovasz:1975}. We then give a
linear complexity iterative rounding algorithm that derives a solution to the
original problem from that of the relaxed version.

We express formally the relaxed optimisation problem as:
\begin{align}
  \small
	&\max_{\mathbf{x}} U := \sum_{j=1}^{M} \sum_{i=1}^{N} x_{i,j}\log{S_{i,j}}, \label{eq:goalrel} \\
	\text{s.t.}\
	&S_{i,j} \leq h_j r_{i,j}, \forall i,j; \label{eq:constrthr}\\
	&\sum_{j=1}^{M} x_{i,j} \leq 1, \forall i; \label{eq:constroneaprel} \\
  	&-x_{i,j} \le 0,  \forall i,j.	\label{eq:constrgezerorel}
\end{align}

\noindent The constraint in (\ref{eq:constroneaprel}) ensures any client $i$ does not
communicate to more than one AP at a given time (single transceiver), while
(\ref{eq:constrthr}) ensures that the throughput allocated to client $i$
when connected to AP $j$ does not exceed the maximum attainable bit
rate under the current signal quality conditions, if client $i$ was the only one
connected to AP $j$. Note that in the original problem (\ref{eq:goal}), where $x_{i,j} \in
\{0,1\}$, it was note necessary to explicitly impose this constraint, since it is implicitly satisfied by (\ref{eq:constroneap}).

\subsection{Convexity Properties and Problem Solution}
Next we analyse the convexity properties of the objective function in the relaxed optimisation problem and give insights into the solution space. 

\begin{lemma}
\label{lem:convexity}
 The utility $U$ function defined by (\ref{eq:goalrel}) is concave.
\end{lemma}

\begin{proof}
 The second order partial derivative of the terms $x_{i,j}\log S_{i,j}$ with respect to $x_{i,j} \neq 0$ is
 \[
   \frac{\partial^2 \left(x_{i,j}\log S_{i,j}\right)}{\partial^2x_{i,j}} = -\frac{1}{\sum_{k=1}^{N} x_{k,j}} - \frac{\sum_{k=1, k \neq
   i}^{N}x_{k,j}}{\left(\sum_{k=1}^{N} x_{k,j}\right)^2} < 0,
 \]
and the same with respect to $x_{l,j} \neq 0, l \neq i$ is
\[
 \frac{\partial^2 \left(x_{i,j}\log S_{i,j}\right)}{\partial^2x_{l,j}} = - \frac{\sum_{k=1, k \neq i}^{N} x_{k,j}}{\left(\sum_{k=1}^{N} x_{k,j}\right)^2} < 0.
\]
Thus the Hessian $\nabla^2\mathbf{x\log(S)}^T$ is negative semi-definite. By Boyd and Vandenberghe \cite{Boyd:2009}, it follows that functions $x_{i,j}\log S_{i,j}$  are concave, and since the utility $U$ is an affine combination of such functions, then it is concave.
 \end{proof}

Since we are working with multiple single-hop TDMA-type systems, the capacity region of which is convex \cite{Clover:2006}, constraint (\ref{eq:constrthr}) is also convex. Further, constraints given by (\ref{eq:constroneaprel}) %(\ref{eq:constrtpfeasrel}), %(\ref{eq:constrleonerel}) 
and (\ref{eq:constrgezerorel})  are convex and thus by Lemma~\ref{lem:convexity} the relaxed optimisation problem defined by (\ref{eq:goalrel})--(\ref{eq:constrgezerorel}) is convex and a solution exists. Slater's sufficient condition is satisfied and thus strong duality holds. The Lagrangian is
\begin{eqnarray*}
 L(\mathbf{x},\mathbf{\lambda},\mathbf{\mu},\mathbf{\nu}) = &-& \sum_{j=1}^{M} \sum_{i=1}^{N} x_{i,j}\log{S_{i,j}} \nonumber \\
  &+& \sum_{i=1}^{N} \sum_{j=1}^{M} \lambda_{i,j} \left(S_{i,j} - h_j r_{i,j}\right) \label{eq:constrfeas}\\
  &+& \sum_{i=1}^{N} \mu_{i} \left(\sum_{j=1}^{M} x_{i,j}-1\right) - \sum_{i=1}^{N} \sum_{j=1}^{M} \nu_{i,j} x_{i,j} \nonumber
\end{eqnarray*}

The Karush-Kuhn-Tucker (KKT) condition \cite{Hillier:2009} for $S_{i,j}$ is
\[
 \frac{\partial L}{\partial S_{i,j}} = 0,
\]
which gives
\[
 \lambda_{i,j} = x_{i,j} \frac{1}{S_{i,j}}.
\]
In the above, we distinguish two possible cases: (1) client $i$ associates to AP
$j$ and thus $x_{i,j} > 0$, which means $\lambda_{i,j} = 0$ (note that $S_{i,j}$
is non zero with $x_{i,j} > 0$); and (2) $x_{i,j} = 0$ from which it follows
that $\lambda_{i,j} = 0$. From complementary slackness it follows that the
inequality constraint (\ref{eq:constrthr}) is not tight and thus the optimum
$\mathbf{x'}$ may not be unique. However we can still employ the widely used
trust region method (TRM) to solve numerically the relaxed optimisation
problem~\cite{TRM}.

\subsection{Rounding Algorithm}
Once we solved the relaxed problem (\ref{eq:goalrel}), the next step is finding a solution to the original utility maximisation problem~(\ref{eq:goal}), which recall is NP-complete. To this end, we design an iterative rounding algorithm that converts the fractional association matrix $\mathbf{x'}$, to an integer association matrix $\mathbf{x^*}$, which is the solution of the original problem (\ref{eq:goal}). 

The simplest way to accomplish this would be a maximum likelihood approach, i.e.
\begin{equation}
  \small
	x_{i,j}^* = \begin{cases}
		1, & \quad \text{if } \max_{k} \{x'_{i,k}\} = x'_{i,j}; \\
		0, & \quad \text{otherwise};
	\end{cases}
\end{equation}
however, this performs poorly in numerous situations, where the solution it returns is identical to that of the SNR-based association. We exemplify this in Fig.~\ref{fig:max-like} for a simple topology with two APs and four stations. Also shown in the figure is the superior performance (87\% higher total throughput) attainable with the iterative rounding algorithm we propose, whose pseudo-code is given in Algorithm~\ref{alg:round_alg} and detailed next.

\begin{figure}[t]
\centering
 \includegraphics[width=0.35\columnwidth]{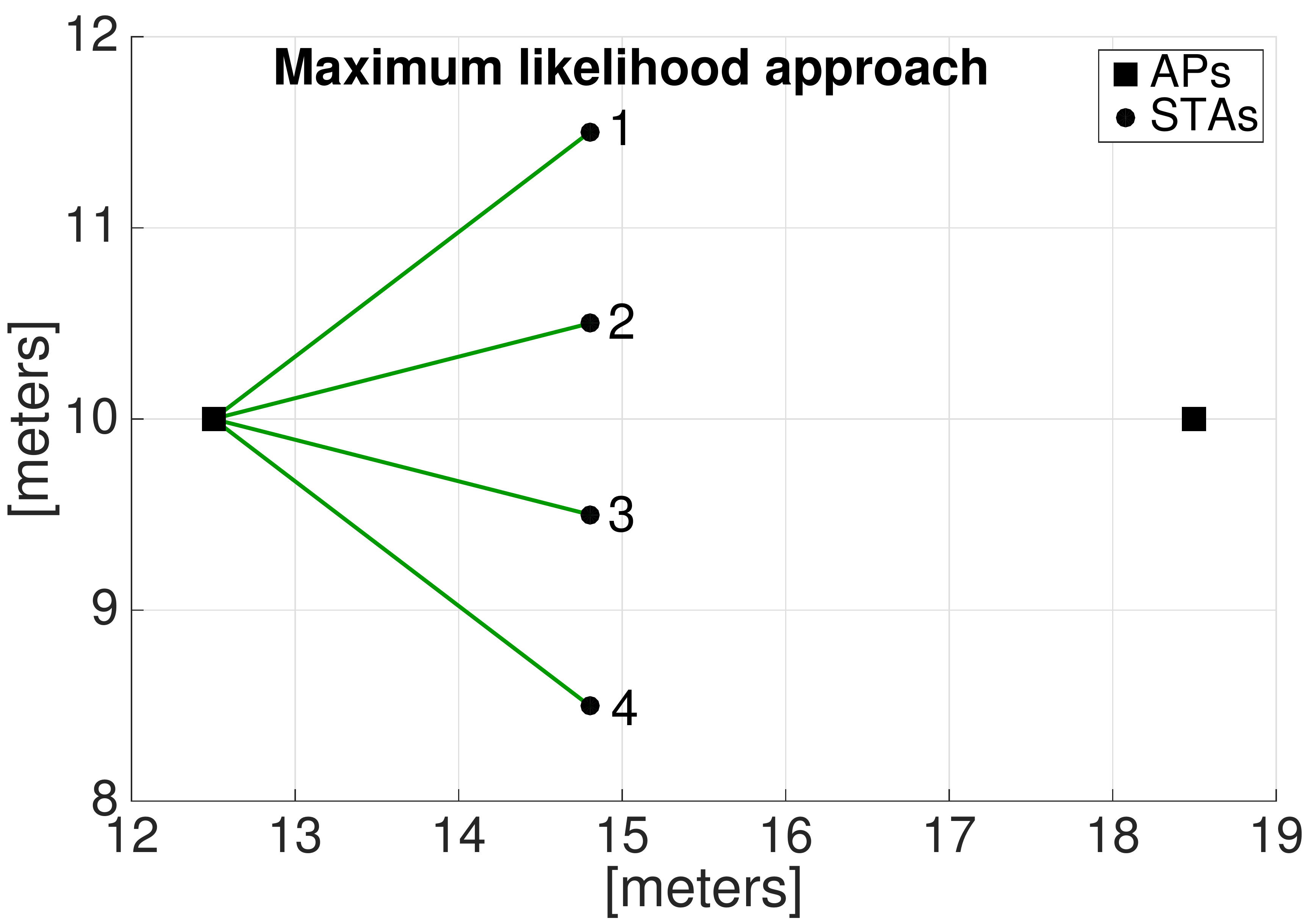}
 \includegraphics[width=0.35\columnwidth]{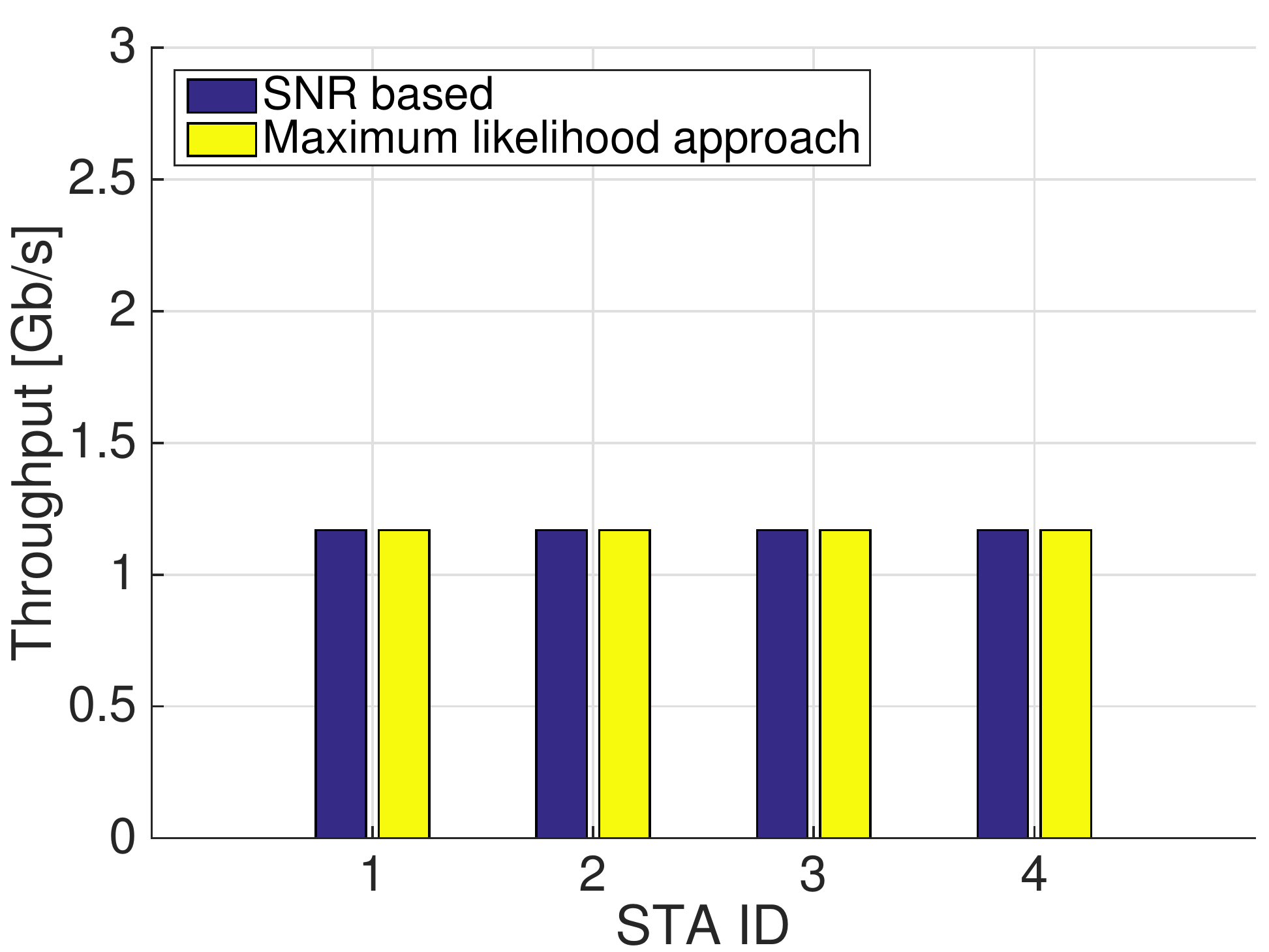}
 \includegraphics[width=0.35\columnwidth]{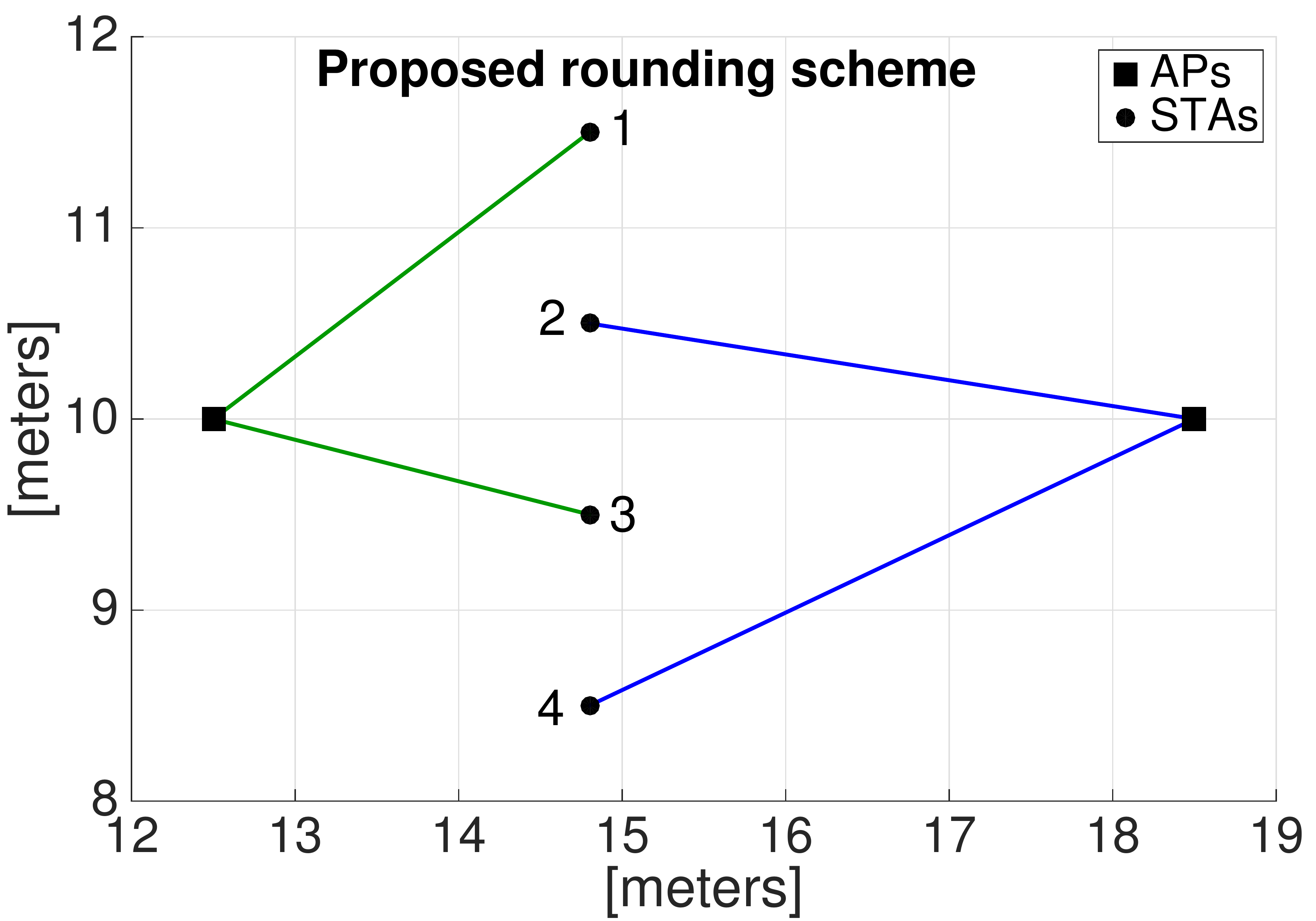}
 \includegraphics[width=0.35\columnwidth]{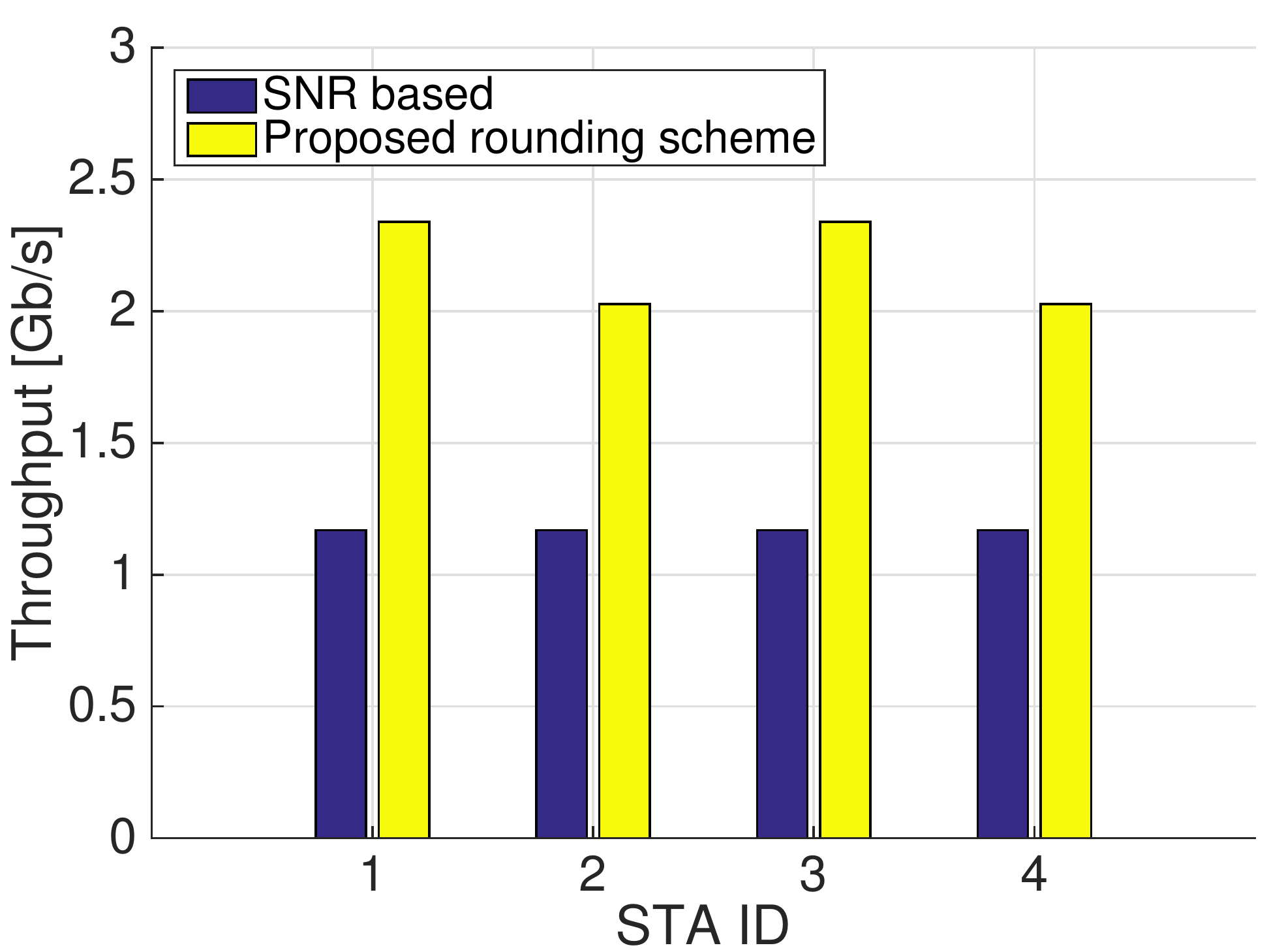}
 \caption{Example of mmWave network with 2 APs and 4 clients. Association enforced by the maximum likelihood approach (top left) and attainable throughput vs that of SNR-based method (top right). Association enforced by the proposed rounding algorithm (bottom left) and corresponding throughput (bottom right). Aggregate throughput gain achieved is 87\%. Numerical results.}
 \label{fig:max-like}
% \vspace*{-0.5em}
\end{figure}

%\vspace*{-0.25em}
\setlength{\textfloatsep}{8pt}
\begin{algorithm}[b!]
%  \small
	\caption{Iterative rounding}%$\mathbf{x^*}=Round(\mathbf{x'}$)}
	\begin{algorithmic}[1]
	\REQUIRE $\mathbf{x'}$ is a feasible solution of problem (\ref{eq:goalrel}).
	\ENSURE $\mathbf{x^*}$ is a feasible solution of problem (\ref{eq:goal}).
	\STATE $X = \{(i,j) \mid 1 \le i \le N, 1 \le j \le M\}$; $X^0 = X$; $n = 0$.
	%\STATE Set $x^*_{i,j} = 0, \forall{i,j}$
	%\FOR{$k = 1$ \TO $k = N$}
	\REPEAT
		%\STATE Find $\hat{i}$, $\hat{j}$ s.t. $x'_{\hat{i},\hat{j}} = \max_{i,j} {x'_{i,j}}$ and $\{\hat{i},\hat{j}\} \in X^n$
		\STATE Find $\hat{i}$, $\hat{j} \text{ s.t. } x'_{\hat{i},\hat{j}} = \max_{(i,j) \in X^n} \left({x'_{i,j}}\right)$;
		\hfill \COMMENT{\emph{{Find $x'_{i,j}$ closest to 1}}}
		%\STATE Set $x^*_{\hat{i},\hat{j}} = 1$
		\STATE Build vector $\mathbf{R}$, s.t. $R_j = x'_{\hat{i}, j}, \forall{j}$;
		\hfill \COMMENT{\emph{{Frac. assoc. freed by $x'_{\hat{i},\hat{j}}$ rounding}}}
		%\STATE Build vector $\mathbf{R}$, s.t. $R_j = x'_{\hat{i}, j} - x^*_{\hat{i}, j}, \forall{j}$;
		\STATE Build vector $\mathbf{D}$, s.t. \\$D_j = \{\|x'_{i,j}\| \mid i \neq \hat{i}, (i,j) \in X^n, r_{i,j} > 0\}, \forall{j}$;
		%$D$ s.t. $D_j = \sharp r_{i \neq \hat{i}, j} \gt 0, \forall{j}$
		\STATE Set $x'_{i,j} = x'_{i,j} + R_j / D_j, \forall{(i,j) \in X^n}$, \\s.t. $i \neq \hat{i}, r_{i,j} > 0$;
		\hfill \COMMENT{\emph{{Update $x'_{i,j}$} not already rounded}}
		\STATE Set $X^{n+1} = X^n \setminus \{(\hat{i}, j) \mid 1 \le j \le M\}$;
		\hfill \COMMENT{\emph{{Remove $(\hat{i},j)$ for rounded $x'_{\hat{i},j}$}}}
		%\STATE Set $x'_{\hat{i}, j} = 0, \forall{j}$
		\STATE Set $x'_{\hat{i}, \hat{j}} = 1$;
		\hfill \COMMENT{\emph{{Round to 1 the selected $x'_{\hat{i},\hat{j}}$}}}
		\STATE Set $x'_{\hat{i}, j} = 0, \forall{j \neq \hat{j}}$;
		\hfill \COMMENT{\emph{{Ensure $\hat{i}$ is associated only to AP $\hat{j}$}}}
		\STATE Set $n = n + 1$;
	%\ENDFOR
	\UNTIL $X^n \neq \varnothing$
		\hfill \COMMENT{\emph{{Rounding complete}}}
	\STATE Set $\mathbf{x^*} = \mathbf{x'}$.
		\hfill \COMMENT{\emph{{$x^*$ is a solution of problem (\ref{eq:goal})}}}
	\end{algorithmic}
	\label{alg:round_alg}
\end{algorithm}

The proposed rounding algorithm requires $N$ iterations (only depending on the number of clients) and thus has linear complexity $O(N)$. We work with a set $X$ which maintains a list of $(i,j)$ tuples corresponding to the $x'_{i,j}$ terms that have not been yet subject to rounding. Initially, $X$ contains all the tuples and we remove $j$ of them at each iteration, as one client is assigned to a \emph{single} AP, until $X$ is empty.

Each iteration is composed of a rounding (lines 3, 8--9) and an update operation
(lines 4--7). The rounding operation sets to 1 (line~8) the $x'_{\hat{i},\hat{j}}$
element whose value is the largest among all $x'_{i,j}$ for which $(i,j)$ is in
$X$ at the current iteration $n$ (line~3). Then, we set to 0 the
$x'_{\hat{i},j}$ terms, $\forall{j \neq \hat{j}}$ (line~9). The key idea is to prioritise
the rounding of the $x'_{i,j}$ closest to 1. If there is more than one such
$(\hat{i},\hat{j})$ tuple, the algorithm chooses one randomly.

The update operation recomputes the values of all the $x'_{i,j}$ which have not
been rounded during iteration $n$ and that are still to be rounded (line
6). The new value is computed by adding to $x'_{i,j}$ the value $R_j /
D_j$, where $R_j$ (line~4) is the fractional association freed on AP $j, j \ne \hat{j}$, by the rounding
of $x'_{\hat{i},\hat{j}}$, and $D_j$ (line~5) is the number of still-to-be-rounded clients
that could associate to AP $j$.
This update is designed in order to satisfy constraint
(\ref{eq:constroneaprel}), i.e.  $\sum_{j=1}^{M} x'_{i,j} = 1, \forall i$.

In Sec.~\ref{sec:evaluation} we demonstrate that by solving the relaxed optimisation problem and subsequently applying our rounding algorithm, we achieve
substantial improvements as compared to SNR-based association control mechanisms for mmWave under saturation condition.
%in more complex network scenarios.
In what follows, we address utility maximisation under finite load circumstances, i.e. when stations do not always have traffic to transmit.

\section{Utility Maximisation for Finite Load Scenario}
\label{sec:finiteload}
In this section we consider the general finite load scenario where each client $i$ has an offered load %defined by the parameter 
$\lambda_i$.
With the introduction of the parameter $\lambda_i$, the definitions of throughput $S_{i,j}$ and airtime $t_{i,j}$ given in (\ref{eq:achtp})
and~(\ref{eq:satairtime}) need to be revisited, since the airtime allocated to client $i$ when associated to the AP $j$ is now also a function of the client's
offered load ($\lambda_i$) and that of the other clients associated to the same AP $j$. This effectively means the airtime $t_{i,j}$ becomes a variable of the optimisation problem, which we formalise as:

\begin{align}
 % \small
	&\max_{\mathbf{x,t}} U := \sum_{j=1}^{M} \sum_{i=1}^{N} x_{i,j}\log{t_{i,j}r_{i,j}}, \label{eq:goalload} \\
	\text{s.t.}\
	&\sum_{j=1}^{M} x_{i,j} = 1, \forall i; \qquad \qquad ~\text{(single AP association)} \label{eq:constroneapload} \\
	&\sum_{i=1}^{N} x_{i,j}t_{i,j} \le h_jT_j, \forall j; \qquad \quad \text{(airtime feasibility)} \label{eq:constratfeasload}\\
	&x_{i,j}t_{i,j}r_{i,j} \le \lambda_i, \forall i,j; \qquad \qquad  ~~\text{(load feasibility)} \label{eq:constrloadfeasload}\\
	&-t_{i,j} \le 0, \forall i,j; \label{eq:constrtimegezeroload}\\
	&x_{i,j} \in \{0,1\}, \forall i,j.	\label{eq:constrdecvarload}
\end{align}

\noindent Finding a solution to the above involves solving two different problems in parallel, namely: 
\begin{enumerate}
 \item Finding the best association matrix $\mathbf{x}$ as in the case of the saturation scenario, and
 \item Finding the best airtime allocation $t_{i,j}$ that takes into account the load requirements $\lambda_i$, while providing some form of fairness.
\end{enumerate}
To accomplish these tasks, we propose an approach that combines simulated annealing and water filling algorithms, and subsequently show that this achieves remarkably higher throughput performance in comparison with the recent DAA scheme \cite{Athanasiou:2015} and the default SNR-based association policy.

\subsection{Simulated Annealing and Water Filling}
\label{sec:siman}
The underlying principle behind solving the problem defined by
(\ref{eq:goalload})--(\ref{eq:constrdecvarload}) is the following:
first we assume saturation conditions and use the method described in
Sec.~\ref{sec:satcond} to find an initial integer association matrix
$\mathbf{x^*}$. We use $\mathbf{x^*}$ as the starting point ($\mathbf{x0}$)
for the simulated annealing algorithm we propose, which we summarise in
Algorithm~\ref{alg:siman} and detail next.
Note that choosing the starting
point in this way ensures a solution is found significantly faster,
as compared to when using the outcome of the SNR-based association instead.
This is particularly true when offered loads $\lambda_i$ are moderate--high,
since $\mathbf{x^*}$ is usually close to the best solution found by our heuristic,
as revealed by analysing multiple topologies.

\begin{algorithm}[t]
%  \small
  \caption{\small{SimulatedAnnealing ($\mathbf{x0}, \mathbf{r}, \mathbf{\lambda}, \mathbf{h}, T0, Tmin, \alpha, q, p$)}}
  \begin{algorithmic}[1]
    \STATE Set $\mathbf{x} = \mathbf{x0}, T = T0, v = 1$ \label{alg:siman:init1}
    \STATE Set $\mathbf{t}=WaterFilling(\mathbf{x}, \mathbf{r}, \mathbf{h}, \mathbf{\lambda})$ \label{alg:siman:init2}
    \REPEAT \label{alg:siman:entertemploop}
    \item[] \# \emph{Stabilisation loop for a given temperature $T$}
    \FOR{$k=1$ \TO $k=q$} \label{alg:siman:enterinnerloop}
      \IF{$x_{i,j}t_{i,j}r_{i,j} = \lambda_i, \forall{i,j} \text{ s.t. } x_{i,j} = 1$} \label{alg:siman:stopcond}
      \STATE Return $\mathbf{x}$ and $\mathbf{t}$ \hfill \COMMENT{\emph{Optimal solution found}}
			\ENDIF
      \STATE Set $\mathbf{x'}=Perturbate(\mathbf{x}, \mathbf{t}, \mathbf{\lambda}, p)$ \label{alg:siman:pert}
      \STATE Set $\mathbf{t'}=WaterFilling(\mathbf{x'}, \mathbf{r}, \mathbf{h}, \mathbf{\lambda})$ \label{alg:siman:wf}
      \STATE Set ${\Delta}E = U(\mathbf{x'}, \mathbf{t'}) - U(\mathbf{x}, \mathbf{t})$ \hfill \COMMENT{\emph{Compute energy}} \label{alg:siman:delta}
			\IF{${\Delta}E > 0$}
        \STATE Set $\mathbf{x} = \mathbf{x'}$, $\mathbf{t} = \mathbf{t'}$ \hfill \COMMENT{\emph{Better solution found}} \label{alg:siman:keep}
			\ELSE
				\STATE Set $y \text{ to a random value} \in [0,1)$
        \IF{$y < e^{{\Delta}E/T}$}
        \STATE Set $\mathbf{x} = \mathbf{x'}$, $\mathbf{t} = \mathbf{t'}$ \hfill \COMMENT{\emph{Accept worse solution}} \label{alg:siman:keepprob}
				\ENDIF
			\ENDIF
		\ENDFOR \label{alg:siman:exitinnerloop}
    \STATE Set $T = T\alpha^v$ \hfill \COMMENT{\emph{Update temperature}} \label{alg:siman:tempdec}
		\STATE Set $v = v + 1$
	\UNTIL $T > Tmin$ \label{alg:siman:exittemploop}
  \STATE Return $\mathbf{x}$ and $\mathbf{t}$
	\end{algorithmic}
	\label{alg:siman}
\end{algorithm}

After the initialisation steps (lines~\ref{alg:siman:init1}--\ref{alg:siman:init2}), the simulated annealing
algorithm enters a loop (lines~\ref{alg:siman:entertemploop}--\ref{alg:siman:exittemploop})
which is executed until the parameter $T$, called temperature, exceeds a
certain minimum $Tmin$.\footnote{We discuss the proposed simulated annealing parameters in Sec.~\ref{sec:evaluation}.} %\footnote{We comment on the precise values chosen for the parameters of the proposed simulated annealing algorithm in Sec.~\ref{sec:evaluation}.} 
The temperature is decremented at each iteration of 
the loop (line~\ref{alg:siman:tempdec}) with a step proportional to a parameter $\alpha$, which controls the 
speed of the algorithm and the granularity of the temperature values ($0 < \alpha < 1$).

For each temperature value $T$, the algorithm enters a second, stabilisation loop
(lines~\ref{alg:siman:enterinnerloop}--\ref{alg:siman:exitinnerloop}),
which explores the solution space (including the initial association matrix $\mathbf{x0}$
and the corresponding airtime allocation computed in line~\ref{alg:siman:init2}).
In line with standard practice, the inner loop is repeated a number of times $q$ proportional to 
the size of the problem the algorithm attempts to solve. In our case
we set $q = \left\lceil{NM/2}\right \rceil$.

Then for every iteration of the stabilisation loop, the algorithm checks if the current solution, given 
by the association matrix $\mathbf{x}$ and the corresponding airtime allocation
$\mathbf{t}$, is able to satisfy the offered load $\lambda_i$ of 
each client $i$ (line~\ref{alg:siman:stopcond}).
If the condition is satisfied, 
the algorithm terminates, returning 
$\mathbf{x}$ and $\mathbf{t}$.
This effectively means the solution is an optimum of
the problem (\ref{eq:goalload})--(\ref{eq:constrdecvarload}) and no other 
solution that does better would be found, given that all offered loads are
satisfied.

If the current solution is not optimal, the algorithm calls a perturbation
function to generate a new association matrix~$\mathbf{x'}$ (line~\ref{alg:siman:pert}).
The perturbation function, whose implementation is domain dependent, is a key component
of simulated annealing, as it 
defines the way in which the solution space is explored. In our case, this
generates a neighbour association matrix~$\mathbf{x'}$ starting from the current one $\mathbf{x}$.
To increase the chances of finding a good solution, the perturbation function must be designed to 
satisfy the \emph{irreducibility property}, i.e. for a number of
iterations that tends to infinity, the starting point of the simulated annealing
should not influence the final result~\cite{siman87}.
As such, %simulated annealing is a memoryless algorithm, in order to satisfy this property,
the perturbation function must introduce some form of randomness when generating a neighbour of
the current solution $\mathbf{x}$. However, it can also employ user-defined
rules to prioritise the generation of particular neighbours over other candidates.
As we will discuss later, for the association problem at hand, we propose a perturbation
function that prioritises the offloading of the bottleneck APs. %, as summarised by Algorithm~\ref{alg:perf}.

Given the new association matrix $\mathbf{x'}$ the algorithm uses the water filling procedure 
in Algorithm~\ref{alg:watfil} to compute the appropriate airtime allocation $\mathbf{t'}$ (line~\ref{alg:siman:wf}).
As we will detail below, this procedure implements an airtime-based water filling 
algorithm, which returns an airtime allocation,\footnote{Note that the
water filling procedure returns a vector $\mathbf{t}$, in which each element
$t_{i,j}$ represents the fraction of super-frame time allocated to client $i$ when
associated to AP $j$. This means the elements of $\mathbf{t}$
returned by Algorithm~\ref{alg:siman} must be translated into actual airtimes,
before being used in our analysis. This is
obtained through a simple conversion, i.e.
$t_{i,j} \leftarrow t_{i,j}(T_j - O_j)$.}
that satisfies the the max-min fairness criterion \cite{Leith:2012}.

Finally, the algorithm computes the energy 
 ${\Delta}E$ (line~\ref{alg:siman:delta}) %= U(\mathbf{x'}, \mathbf{t'}) - U(\mathbf{x}, \mathbf{t})$ 
as the difference between the utility obtained using the new
solution $(\mathbf{x'}, \mathbf{t'})$ and respectively
the previous one, according to (\ref{eq:goalload}). 
If the energy is positive, this means the new solution is better %than the previous 
and thus must be kept (line~\ref{alg:siman:keep}). Otherwise, we keep the new
solution only with a probability $e^{{\Delta}E/T}$ that depends on the current 
energy ${\Delta}E$ and temperature $T$.

The water filling based airtime allocation procedure invoked at line~\ref{alg:siman:init2} is outlined in Algorithm~\ref{alg:watfil}. First, this
computes the airtime $t^{\lambda}_{i,j}$ required to satisfy the
load $\lambda_i$ for each client $i$ and each AP $j$ (line~\ref{alg:watfil:reqair}),  
and then loops over all the APs (lines~\ref{alg:watfil:enteraploop}--\ref{alg:watfil:exitaploop}) as follows.
Three sets, $A_j$, $A_{j}'$, and $A_{j}''$ (line \ref{alg:watfil:defsets}) maintain 
the list of clients that must be associated to AP $j$ (i.e. $x_{i.j} = 1$) and whose corresponding
airtimes $t_{i,j}$ have not been set yet; the list of clients with allocated airtimes; and the list of clients who can only be allocated a fraction of the
airtime required to satisfy their load. In addition, the residual airtime still available at AP $j$ is maintained in $\hat{h}_j$ (line~\ref{alg:watfil:res}). 

\begin{algorithm}[!h]
%  \small
  \caption{WaterFilling ($\mathbf{x}, \mathbf{r}, \mathbf{h}, \mathbf{\lambda}$)}
	\begin{algorithmic}[1]
    \STATE Set $t_{i,j} = 0, \forall{j}$ and $f = 0$ \label{alg:watfil:init}
    \STATE Set $t^{\lambda}_{i,j} = h_j \lambda_{j}/r_{i,j}, \forall{i,j}$ \label{alg:watfil:reqair}
    \FOR[\emph{Loop on APs}]{$i=1$ \TO $i=M$} \label{alg:watfil:enteraploop}
    \STATE Define $A_j = \{i \mid x_{i,j} = 1, \forall{i}\}$, $A_{j}' = \varnothing$ and $A_{j}'' = \varnothing$ \label{alg:watfil:defsets}
    \STATE Set $\hat{h}_{j} = h_j$ \label{alg:watfil:res} \hfill \COMMENT{\emph{AP $j$ residual time}}
    \REPEAT[\emph{Loop on clients associated to AP $j$}] \label{alg:watfil:enterinner}
    \STATE Set $f = \hat{h}_j / (\|A_{j}\| + \|A_{j}''\| - \|A_{j}'\|)$ \label{alg:watfil:frac}
  \item[] \hspace*{-1em}\{\emph{Try to satisfy load requested by the "easier" client $\hat{i}$}\}
    \STATE Set $\hat{i} = i$, s.t. $t^{\lambda}_{i,j} = \min_{i \in A{j}} {t^{\lambda}_{i,j}}$
    \IF[\emph{Client load can be satisfied}]{$t^{\lambda}_{\hat{i},j} < f$} \label{alg:watfil:suc1}
      \STATE Set $t_{i,j} = t^{\lambda}_{\hat{i},j}$, $\hat{h}_{j} = \hat{h}_{j} - t^{\lambda}_{\hat{i},j}$
      \STATE Set $A_{j} = A_{j} \setminus \{\hat{i}\}$, $A_{j}' = A_{j}' \cup {\{\hat{i}\}}$ \label{alg:watfil:suc2}
    \ELSE[\emph{Client load can not be satisfied}] \label{alg:watfil:fail1}
      \STATE Set $A_{j} = A_{j} \setminus \{\hat{i}\}$, $A_{j}'' = A_{j}'' \cup {\{\hat{i}\}}$ 
    \ENDIF \label{alg:watfil:fail2}
    \UNTIL $A_{j} \ne \varnothing$ \label{alg:watfil:exitinner}

    \STATE Set $t_{i,j} = f, \forall{i \in A_{j}''}$ \COMMENT{\emph{Set time $f$ for unsatisfied clients}}\label{alg:watfil:equal} 

  \ENDFOR \label{alg:watfil:exitaploop}
	\STATE Return $\mathbf{t}$
	\end{algorithmic}
	\label{alg:watfil}
\end{algorithm}

Then an inner loop (lines~\ref{alg:watfil:enterinner}--\ref{alg:watfil:exitinner}) first computes the 
fraction of equal airtime $f$ that can be assigned to each client $i$ (line \ref{alg:watfil:frac}) and 
selects from $A_j$ the index of the client $i$ whose corresponding $t^{\lambda}_{i,j}$ is the minimum among all set members; 
i.e. it searches the client whose load request is the easiest to satisfy.
If the time required to satisfy client $i$'s load ($t^{\lambda}_{i,j}$)
is less than the fraction of airtime available to that client ($f$), it means  
AP $j$ can completely satisfy that request (lines~\ref{alg:watfil:suc1}--\ref{alg:watfil:suc2}). Therefore  
the airtime allocated to client $i$ associated to AP $j$ is set to $t^{\lambda}_{i,j}$, the residual time
 $\hat{h}_j$ available at AP $j$ is updated, the current index $i$ is removed from set $A_j$ and 
 inserted in $A_{j}'$.
If instead the fraction of available airtime ($f$) is insufficient to satisfy the 
load request, the current index $i$ is removed from $A_j$ and 
 inserted in $A_{j}''$ (lines~\ref{alg:watfil:fail1}--\ref{alg:watfil:fail2}).
 Finally, an equal slice of the residual airtime is assigned to
each client in $A_{j}''$ (line~\ref{alg:watfil:equal}).

\begin{algorithm}[!b]
%  \small
  \caption{$\mathbf{x'}=Perturbate(\mathbf{x}, \mathbf{t}, \mathbf{\lambda}, p)$}
	\begin{algorithmic}[1]
	\STATE Let $\mathbf{x'} = \mathbf{x}$
    \STATE Builds sets $B^-$ and $B^+$ as defined by (\ref{eq:setbn}).
	\STATE Set $y \text{ to a random value} \in [0,1)$
    \IF[\emph{irreducibility property}]{$y < p$} \label{alg:perf:rand1}
    \STATE Choose random $(i,j), \text{ s.t. } x_{i,j} = 1$
    \STATE Choose random  $\hat{j}, \text{ s.t. } \hat{j} \neq j, r_{i,\hat{j}} > 0$ \label{alg:perf:rand2}
    \ELSE[\emph{Bottlenecks offloading}]
    \IF{$B^- \neq \varnothing$} \label{alg:perf:nb1}
      \STATE Choose random $(i,j), \text{ s.t. } x_{i,j} = 1, j \in B^+$
      \STATE Choose random  $\hat{j}, \text{ s.t. } \hat{j} \in B^-, r_{i,\hat{j}} > 0$ \label{alg:perf:nb2}
    \ELSE
      \STATE Choose random $(i,j), \text{ s.t. } x_{i,j} = 1, B_j \neq \min_{j^*} B_{j^*}$\label{alg:perf:b1}
			\STATE Choose random  $\hat{j}, \text{ s.t. } r_{i,\hat{j}} > 0, B_{\hat{j}} < B_j$ \label{alg:perf:b2}
		\ENDIF
	\ENDIF
	\STATE Set $x'_{i,j} = 0$
	\STATE Set $x'_{i,\hat{j}} = 1$
	\end{algorithmic}
	\label{alg:perf}
\end{algorithm}

We conclude this section with a brief description of the perturbation function summarised in
Algorithm~\ref{alg:perf}, %~\footnote{To simplifying notation, in Perturbate function we assume always that 
%$ 1 \le i \le N, 1 \le \hat{i} \le N, 1 \le j \le M$ and $1 \le \hat{j} \le M$.}
whose key objective is to prioritise the offloading of the network bottlenecks.
Given an association matrix $\mathbf{x}$ and
an airtime allocation matrix $\mathbf{t}$, we define the bottleneck value $B_j$ for each AP $j$ as:
\begin{equation}
  \label{eq:bj}
  \begin{split}
   B_j = B^{load}_j - B^{time}_j
  \end{split}
\end{equation}
where
\begin{equation}
  \label{eq:bload}
  B^{load}_j = \sum_{i=1}^N x_{i,j}\lambda_i - \sum_{i=1}^N x_{i,j}r_{i,j}t_{i,j}
\end{equation}
\begin{equation}
  \label{eq:btime}
	%\nick{B^{time}_j = h_j - \sum_{i=1}^N x_{i,j}r_{i,j}t_{i,j}}
	B^{time}_j = (T_j - O_j) - \sum_{i=1}^N x_{i,j}r_{i,j}t_{i,j}
\end{equation}

\noindent $B^{load}_j \ge 0$ is the difference between the total load request of clients
$i$ associated to AP $j$ and the amount of requested load AP $j$ is able to
satisfy. We have that $B^{load}_j = 0$ when the AP can completely satisfy the
load request. Instead, $B^{time}_j \le 0$ is the difference between the total
airtime available in a super--frame at AP $j$ and the airtime consumed by
the associated clients.  We also have that $B^{time}_j = 0$ when $B^{load}_j >
0$ and that $B^{time}_j \ge 0$ when $B^{load}_j = 0$.  From
(\ref{eq:bj})--(\ref{eq:btime}) it is easy to observe that AP $j$ is a network
bottleneck when $B_j \ge 0$ and it is not when $B_j < 0$. The perturbation function starts by building the following sets:
\begin{align}
\label{eq:setbn}
\begin{aligned}
  B^- = \{j \mid B_j < 0\}, \forall{j};  \\
  B^+ = \{j \mid B_j \ge 0\}, \forall{j},
\end{aligned}
\end{align}
where $B^-$ and $B^+$ contain indexes $j$ of the APs that are not, and respectively are bottlenecks. 
 
Then the algorithm selects a random client $i$ with probability $p$, and moves it to a different AP (lines~\ref{alg:perf:rand1}--\ref{alg:perf:rand2}).
This introduces the randomness required for satisfying the \emph{irreducibility property}.
On the other hand, the algorithm tries to reduce network bottlenecks with a probability $1-p$,
%If APs that are not bottlenecks exist, the algorithm moves 
moving a random client from 
a bottlenecked AP to one with available resources (lines~\ref{alg:perf:nb1}--\ref{alg:perf:nb2}).
%Instead, if 
If all the APs are bottlenecks, the algorithm moves a random client from an AP~$j$ to a different AP $\hat{j}$ with $B_{\hat{j}} < B_j$. %(lines~\ref{alg:perf:b1}--\ref{alg:perf:b2}). 

\section{Performance Evaluation}
\label{sec:evaluation}

In this section we evaluate the performance of the proposed association control
algorithms under different traffic load conditions, considering enterprise environments
where clients are within coverage of multiple APs and encounter different link
qualities. We compare the performance of our solutions in terms of individual and total throughput, as well as network utility, with that of the 802.11ad
standard SNR-based policy, \rev{that of a greedy association algorithm whereby APs take turns in associating the nearest clients,} and respectively that of the recent DAA scheme~\cite{Athanasiou:2015}.
We further compare the performance of
all approaches in small topologies, with that of the global optimum we obtain through
exhaustive search. Subsequently, we evaluate the average performance attained by the proposed and existing schemes, when stations' position evolve according to a random
waypoint mobility model. Lastly, we demonstrate the short runtime of the
combined simulated annealing and water filling mechanism we propose.

\subsection{Simulation Environment}
\label{sec:simenv}

As mmWave platforms suitable for large scale experimentation are yet to appear
\cite{Sur:2015}, to evaluate our proposal we develop an NS-3\footnote{NS-3
discrete-event network simulator, {\ttfamily{https://www.nsnam.org/}}}
simulation module that implements closely the 802.11ad protocol
details~\cite{802.11ad}.
We employ the isosceles cone antenna pattern defined in~\cite{geofencing}, which
can be steered to an arbitrary angle and whose elevation and azimuth are
functions of the gain.  In our simulations we configure the antenna gain to
15dB,
%which corresponds to a $\sim$41$^{\circ}$ beam width.\footnote{Modern 60GHz
%antenna arrays will be expected to form beams of $\sim$6$^{\circ}$ width or
%less, making the \emph{pseudo-wired} assumption practical~\cite{flyways}.} For
%what regards the antenna, which is a key aspect for ensuring the validity of
%the \emph{pseudo-wired} point-to-point connections assumption, we adopt the
%model used in the Geo-fencing project~\cite{geofencing}: the antenna pattern is
%modelled as an isosceles cone that can be steered to an arbitrary angle and
%whose elevation and azimuth are functions of the gain. For our simulations we
%configured the antenna gain to 15dBm with corresponding beamwidth\footnote{Note
%that modern 60GHz antenna technologies might in practive have a granularity of
%$\sim 6^{\circ}$ or less, making the \emph{pseudo-wired} assumption even easier
%to satisfy~\cite{flyways}.} of $\sim 41^{\circ}$.
and note this model was used successfully by Halperin \emph{et al.}, who also
developed a basic NS-3 implementation of the 802.11ad physical layer
\cite{flyways}. We used their code as a starting point for our own
implementation, adding the missing Directional Multi-Gigabit (DMG) PHY
capabilities and the scheduled Service Period (SP) based MAC within the Data Transmission Interval (DTI), as
illustrated in Fig.~\ref{fig:superframe}. Similarly to~\cite{flyways},
we compute the SINR for different parts of the frames, combining the power from
multiple interferers and noise, and model free space propagation using Friis
law. \rev{We consider indoor deployments with ceiling mounted mmWave APs, in which simulation results we obtain reveal that the collision rates are below 0.003. This confirms the validity of the pseudo-wired link assumption used.}
To estimate the Bit Error Rate (BER), we used the receiver sensitivity
specified by the standard (table 21-3 in~\cite{802.11ad}). \rev{While advanced channel modelling is outside of this work, the assumptions we make are appropriate for indoor scenarios with finite coverage and number of access points. By allowing for arbitrary SNRs on client--AP links, we decouple the client association and airtime allocation tasks at the core of our contribution, from the environment dependent (e.g. reflections, obstacles, etc.) PHY channel properties already documented~\cite{Saha2017,Nitsche:2015}.}

% which indicates the sensitivity for each rate and coding as the (SINR) power
% level down to which a device much successfully receive more than 99\% of
% 4096-byte frames sent using that rate.
%During simulations all nodes employ OFDM PHY rates. %and 1mW transmit power.

%As in~\cite{flyways} we used the standard NS-3 modelling technique for
%computing the SINR levels for different parts of the frames. This simple SINR
%model, which adds together the power from multiple interferers, combines it
%with noise, and compares it with signal strength, has been found appropriate
%with directional antennas~\cite{dirc}. Then, for estimating the Bit Error Rate
%(BER), we used the receiver sensitivity specified by the table 21-3 of the
%802.11ad standard~\cite{802.11ad} which indicates the  sensitivity for each
%rate and coding as the (SINR) power level down to which a device much
%successfully receive more than 99\% of 4096-byte frames sent using that rate.
%Finally, we modelled the signal free space propagation using the Friis' law and
%during simulations all the nodes were configured for using the faster OFDM PHY
%MCSs and a transmission power of 1mW.

%In terms of MAC protocol implementation, we develop the SP based channel access mechanism within the Data Transmission Interval (DTI), as illustrated in Fig.~\ref{fig:superframe}.  
Recall that an SP is
allocated for contention-free access between a client and an AP,
%The node identified as the source of the service period accesses the channel
%separating frames by $3{\mu}s$ (a SIFS) and
without carrier sensing. We implement A-MPDU aggregation for efficient
transmission (up to 64 frames in a single A-MPDU) and the Block Ack mechanism.
%We didn't implement the beaconing and the antenna training protocols but we
%took into account their overhead by forcing a silence period of ${O_j}$ seconds
%in each beacon interval. In particular, for our simulations we set $T_j = 102.4
%ms$ and $O_j = 0.1{T_j}$ for each $j$. We also ignored the antenna pattern
%switching overhead and we assumed  the capability of the sender and transmitter
%to steer their antenna beams for pointing at each other.

Since the standard does not specify a beamforming training mechanism, we use a
conservative 10\% overhead for this procedure, noting that performance gains
will remain unchanged with other values, and assume negligible beam switching
overhead.  The standard neither mandates a specific rate control algorithm,
therefore we implement a rate controller that selects the best transmission MCS
based on the SINR measured at the receiver. We assume that each clients $i$ can
estimate the rates $r_{i,j}$ towards each AP~$j$ by measuring the SINR of the
beacons the APs transmit periodically on each antenna sector.  Finally, we
extend \mbox{NS-3} to enable automatic generation of network topologies and
rapid configuration of 802.11ad WLANs.\footnote{\rev{The source code of our NS-3
simulation module is publicly available at {\ttfamily
\url{https://bitbucket.org/uoeunibs/11ad-for-ns3}}}}

%but these components are straightforward and we omit them due to lack of space.

% \nick{In the simulations we consider two scenarios: the first one is an indoor
% 24$\times$20 indoor deployment where 4 ceiling mounted mmWave APs are placed in
% a square layout, as depicted by the black squares in Fig.~\ref{fig:pmf4ap10sta}.
% In this area 10 client stations are distributed extracting their positions using
% the probability mass functions reported in the same figure. We simulated in
% total 30 different deployments of the 10 client stations, each one repeated
% three times, and the results reported in the following sections referring to
% this scenarios are the averages computed over the 30 deployments. The black
% circles in Fig.~\ref{fig:pmf4ap10sta} report an example of client deployments
% used in simulations. In this scenario the clients transmit at PHY rates between
% 693Mb/s and 6.756Gb/s, depending on their relative distance to the APs within
% range~\cite{802.11ad}. Finally, for this scenario, given the small amount of APs
% and client stations we are able to compute also the optimal solution of problems
% \ref{eq:goal} and \ref{eq:goalload} using an exhaustive search approach.
% }
\begin{figure}[t]
\centering
 \includegraphics[width=0.6\columnwidth]{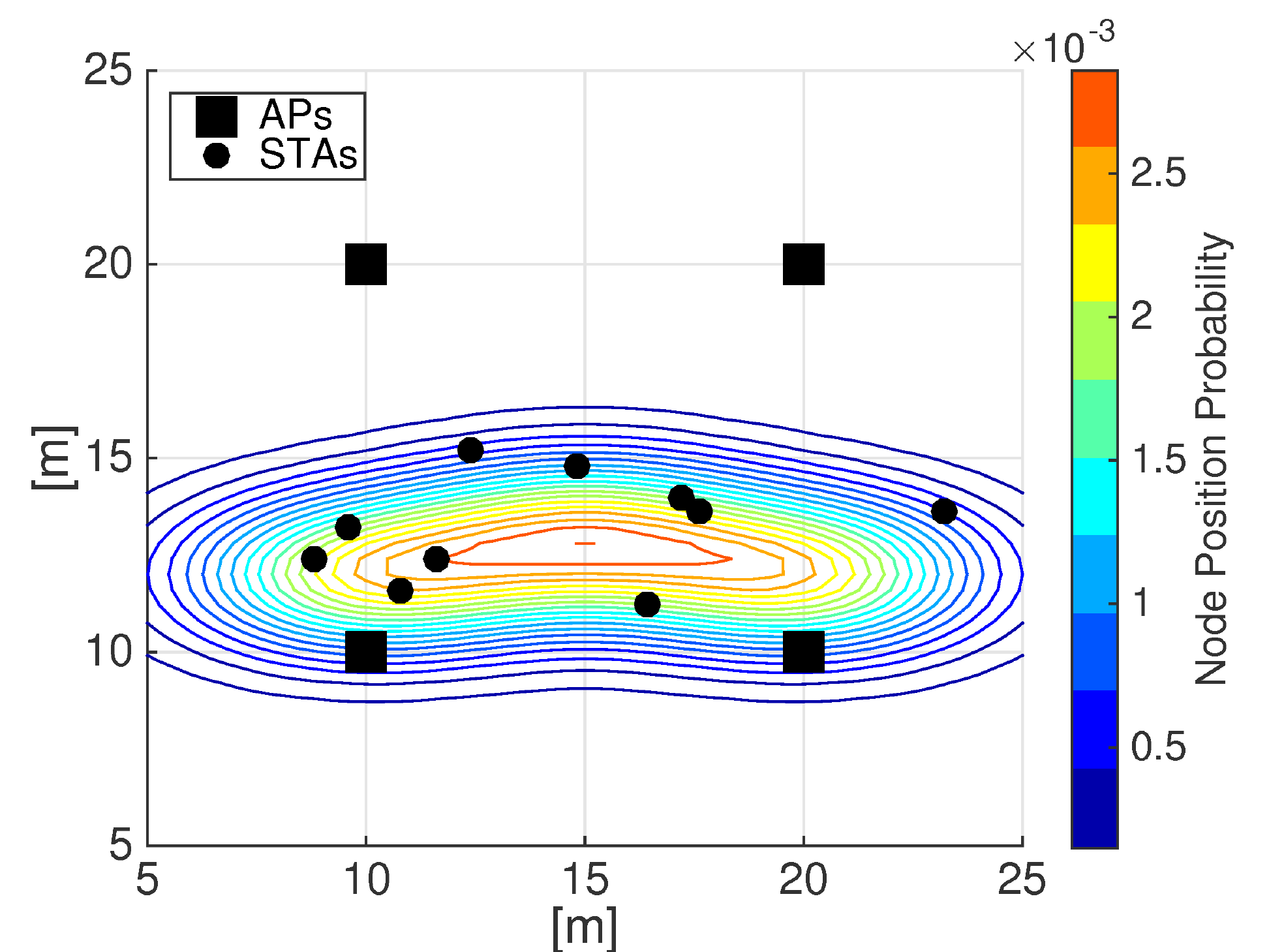}

 \caption{Probability mass function used for extracting the positions of the
	stations in 4APs/10STAs scenarios used for evaluation. Probability
	increases from blue ($\approx0.2 \cdot 10^{-3}$) to red ($\approx2.9
	\cdot 10^{-3}$).}

 \label{fig:pmf4ap10sta}
\end{figure}

For evaluation purposes we consider two deployment scenarios. The first is an
indoor 24m$\times$20m area, where four ceiling mounted mmWave APs are placed in
a square layout, as depicted by the black squares in Fig.~\ref{fig:pmf4ap10sta}.
Ten client stations are randomly distributed by extracting their positions using
the probability mass function shown
as a contour plot in the same figure. The probability decreases from red
to blue, with the maximum probability of $\approx2.9\cdot10^{-3}$ centered at
coordinate $(15, 13)$, and the edge contour line corresponding to a
probability of $\approx0.2\cdot10^{-3}$. Note that the probability
is never zero and there is a low chance to extract positions outside the outer contour line.

\begin{table}
	\center
	\scriptsize
\begin{tabular}{ |l|l| }
	\hline
	\multicolumn{2}{|c|}{\textbf{Simulation Environment}} \\
	\hline
	        Antenna Model & Cone pattern \\
	\hline
	        Antenna gain & 15 dBm \\
	\hline
		Antenna beamwidth & $\sim 41^{\circ}$ \\
	\hline
		Channel access & Service Periods based\\
	\hline
		Propagation model & Free space (Friis law)\\
	\hline
		Bit Error Rate &  Receiver sensitivity (table 21-3 in~\cite{802.11ad})\\
	\hline
		A-MPDU aggregation & 64 frames\\
	\hline
		DMG PHY & OFDM up to 6.756Gb/s \\
	\hline
		Beamforming training overhead & 10\% \\
	\hline
		Rate controller & SNR-based \\
	\hline
		Traffic type & 1470-byte UDP downlink packets \\
	\hline
		\multirow{3}{*} {Evaluated scenarios} & Backlogged Traffic \\
		& Finite Load Conditions \\
		& User Mobility \\
	\hline
		\multirow{3}{*} {Considered metrics} & Individual throughput \\
		& Total throughput \\
		& Network utility \\
	\hline
		\multirow{3}{*} {Evaluated algorithms} & Proposed solutions \\
		& SNR-based (Eq. and Water fill airtime) \\
		& DAA scheme \\
	\hline
		\multirow{2}{*} {Deployments} & 24m×20m area, 4 APs, 10 clients \\
		& 30m×30m area, 9 APs, 30 clients \\
	\hline

\end{tabular}
	\caption{Simulation environment summary}
	\label{tab:simenv}
\end{table}

To obtain average results of
measured throughputs with good statistical significance, we consider a total of
30 different deployments of this type, and compute the average individual
throughput in each case over three simulation runs (i.e. 90 simulations in
total). The black circles in Fig.~\ref{fig:pmf4ap10sta} show an example of
client locations used in simulation. The clients transmit at PHY rates between
693Mb/s and 6.756Gb/s, depending on their relative distance to the APs within
range~\cite{802.11ad}. Given the small number of APs and stations in this first
scenario, we will also compute the optimal solution of the problems
(\ref{eq:goal}) and (\ref{eq:goalload}) using exhaustive search.

% \nick{In the second scenario} we consider an indoors 30m$\times$30m
% deployment where 9 ceiling mounted mmWave APs are placed on a grid layout,
% as depicted in Fig.~\ref{fig:SNRassoc}. In this area 30 client stations are
% distributed randomly and transmit at PHY rates %that vary
% between 1.7325 Gb/s and 6.237 Gb/s, depending on their relative distance to the
% APs within range~\cite{802.11ad}.
% \nick{For this scenario, the results we report are the averages computed over
% ten repetition of the simulation. Due to the number of APs and client stations
% in this topology it was not possible to use exhaustive search for computing the
% optimal solutions.}

In the second scenario we consider a more complex 30m$\times$30m
indoors deployment, where 9 ceiling mounted mmWave APs are placed on a grid layout,
as depicted in Fig.~\ref{fig:SNRassoc}, and serve 30 randomly placed client stations. 
Here, the results we report are the averages computed over ten repetition of the simulation. Given the number of APs and client stations
in this topology, finding the absolute optimum through exhaustive search is no longer feasible. In all scenarios APs transmit 1470-byte UDP packets in the downlink. Table~\ref{tab:simenv} summarises the simulation parameters and scenarios we consider for evaluation.
%transmission.

\begin{figure*}[!t]
\centering
	\begin{subfigure}[t]{0.74\textwidth}
	\centering
	\includegraphics[width=\textwidth]{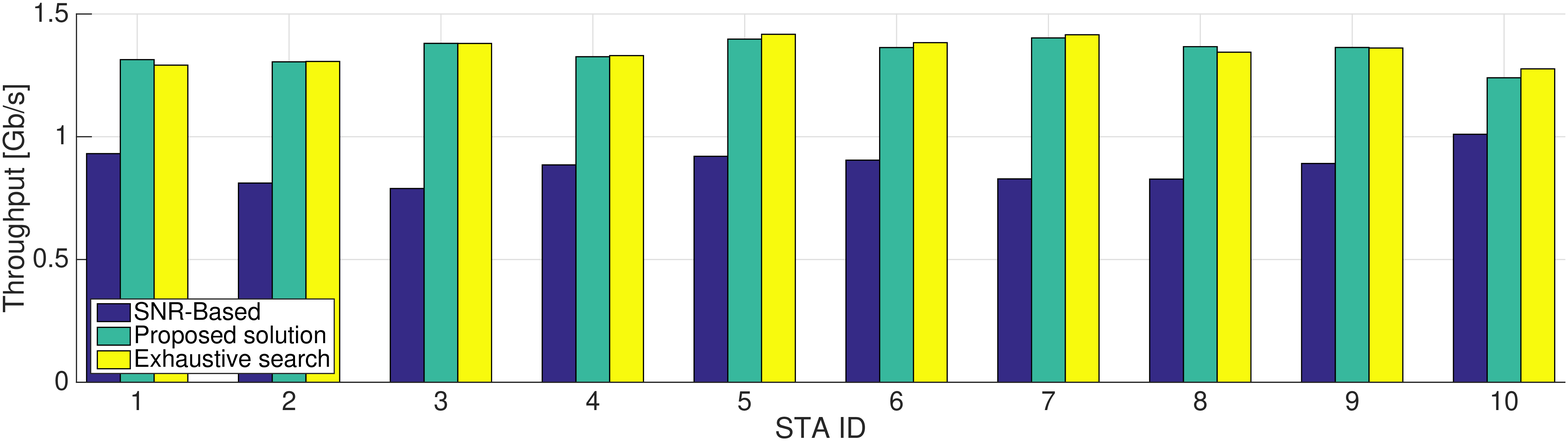}

	\caption{Individual throughputs achieved with the SNR-based and
		proposed methods.}

	\label{fig:SatPerf_4ap10sta}
	\end{subfigure}
	\begin{subfigure}[t]{0.25\textwidth}
	\includegraphics[width=\textwidth]{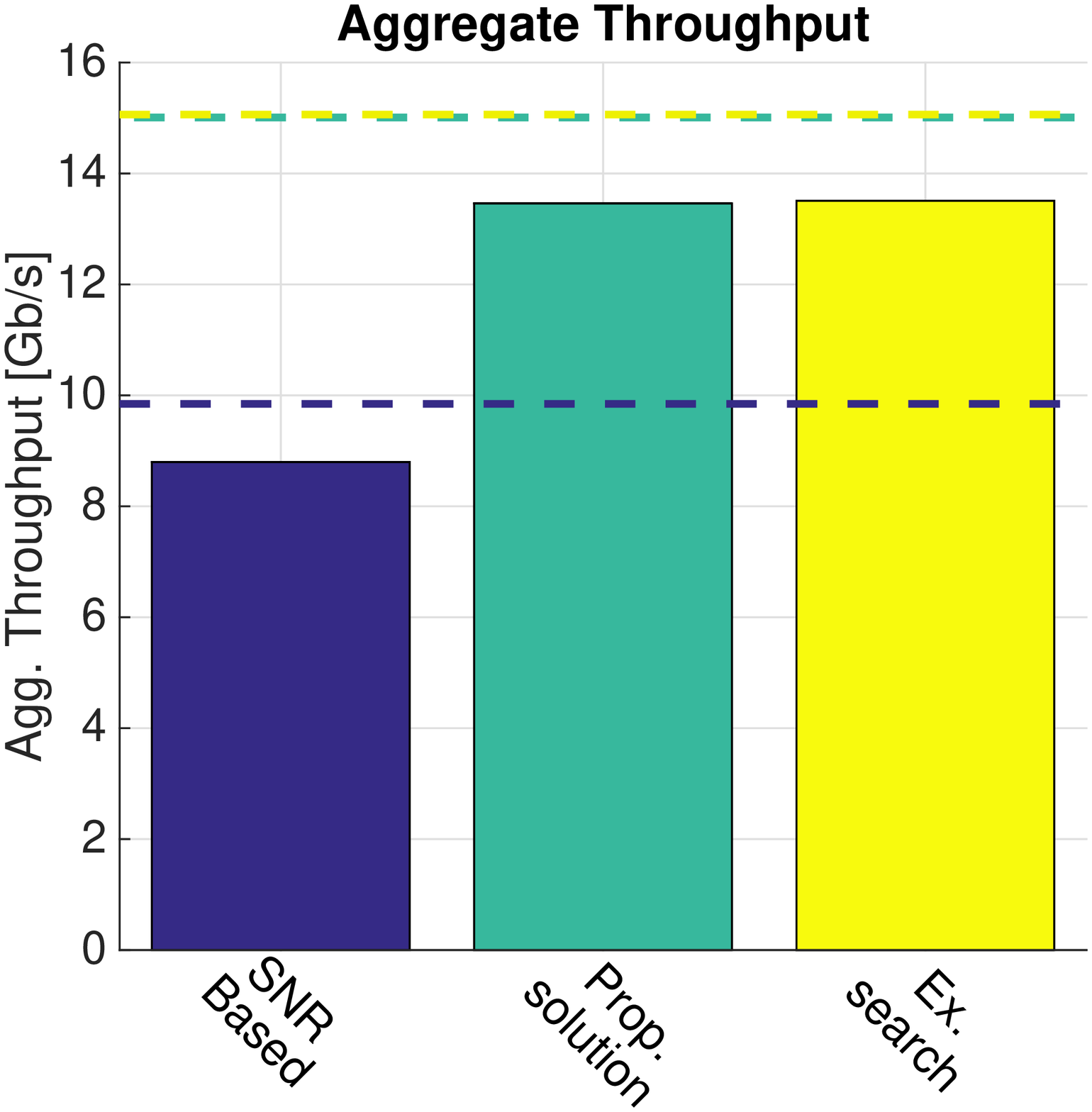}
		\caption{Aggregate~throughput comparison.}
	\label{fig:agg_tp_sat_4ap10sta}
	\end{subfigure}

	\caption{Enterprise mmWave network with 4 APs positioned on a grid
	and 10 client stations deployed using the pmf in
	Fig.~\ref{fig:pmf4ap10sta}.  Client throughput performance attained with
	the SNR-based policy, the proposed association control mechanism and
	respectively exhaustive search. All clients are backlogged (saturation
	conditions) and equal airtime allocation is performed at each AP. Theoretical maximum shown with dashed lines.
	Simulation results.}

\label{fig:Sat_4ap10sta}
% \vspace*{-1.25em}
\end{figure*}

\begin{figure*}[!ht]
\begin{adjustwidth}{-1cm}{-1cm}
%\hspace*{-1.5in}
\centering
	\begin{subfigure}[t]{0.35\textwidth}
	\centering
	\includegraphics[width=\textwidth]{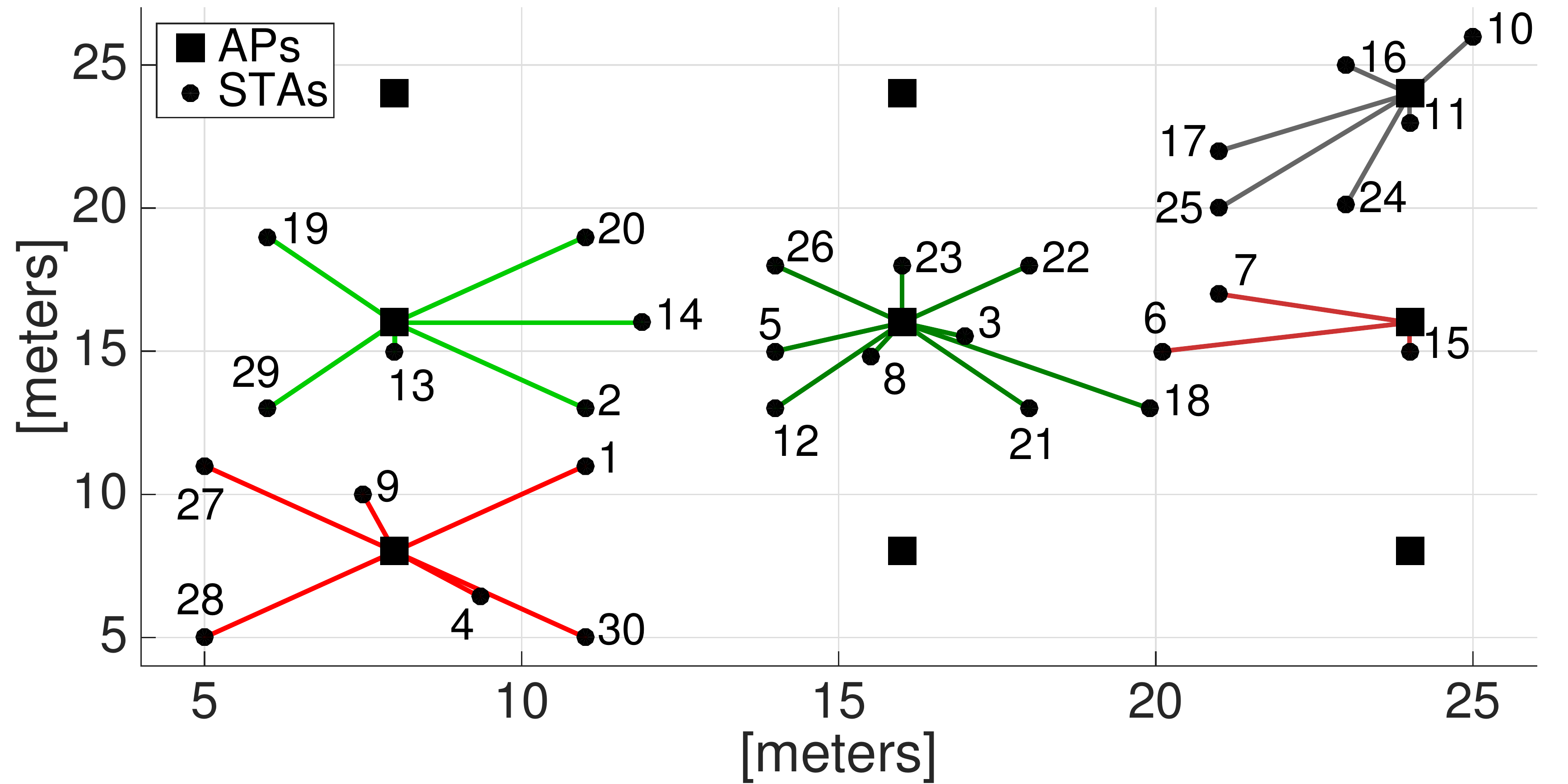}
	\caption{Client--AP links established using highest SNR policy.}
	\label{fig:SNRassoc}
	\end{subfigure}
	\hspace*{0.45em}
	\begin{subfigure}[t]{0.35\textwidth}
	\centering
	\includegraphics[width=\textwidth]{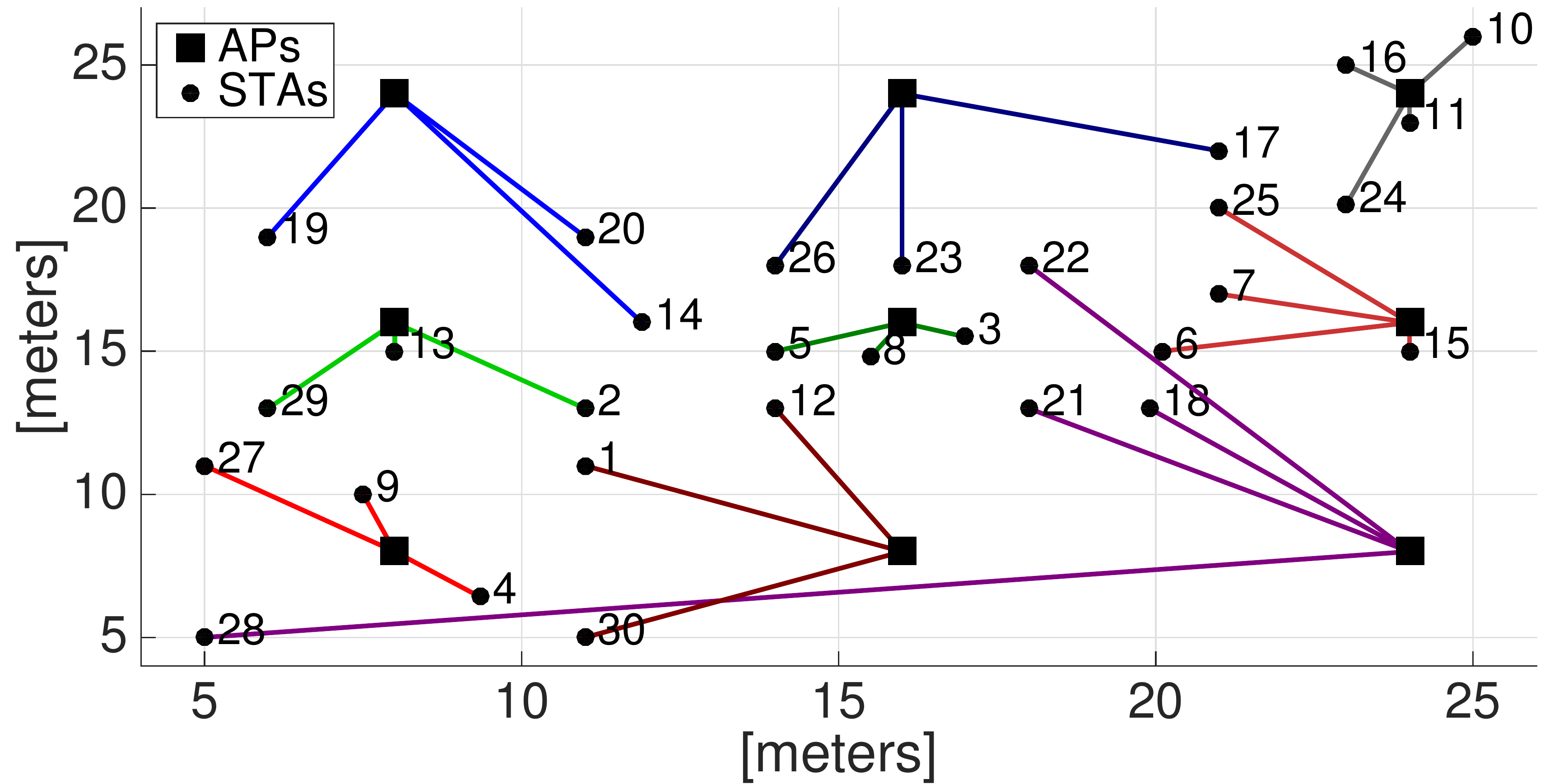}
	\caption{\rev{Client--AP links obtained via greedy association.}}
	\label{fig:GreedyAssoc}
	\end{subfigure}	
	\hspace*{0.45em}
	\begin{subfigure}[t]{0.35\textwidth}
	\centering
	\includegraphics[width=\textwidth]{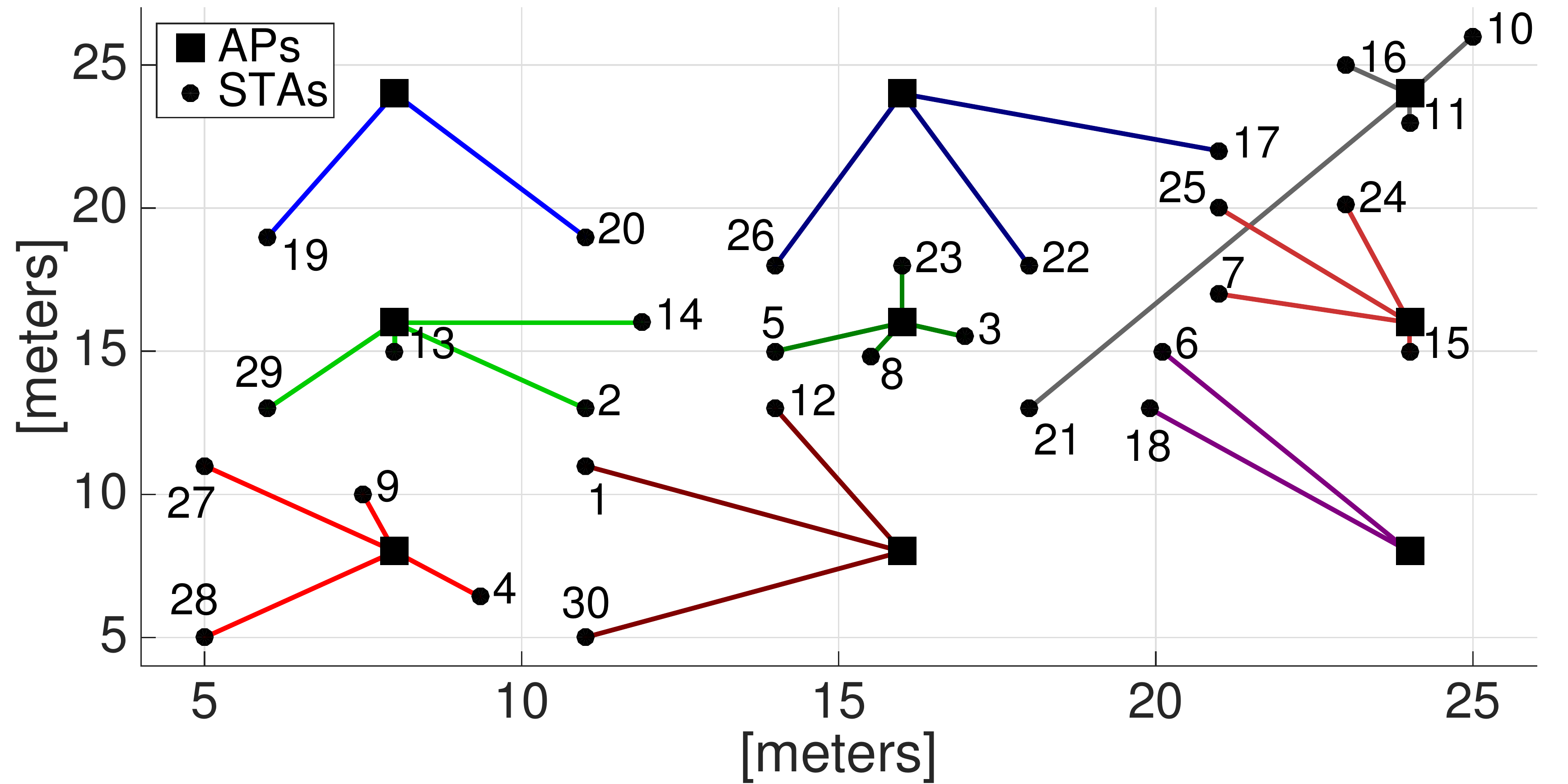}
	\caption{Client--AP links enforced by the proposed association algorithm.}
	\label{fig:SatAssoc}
	\end{subfigure}	
%\vspace*{1em}
	\begin{subfigure}[t]{.85\textwidth}
	\centering
	\includegraphics[width=\textwidth]{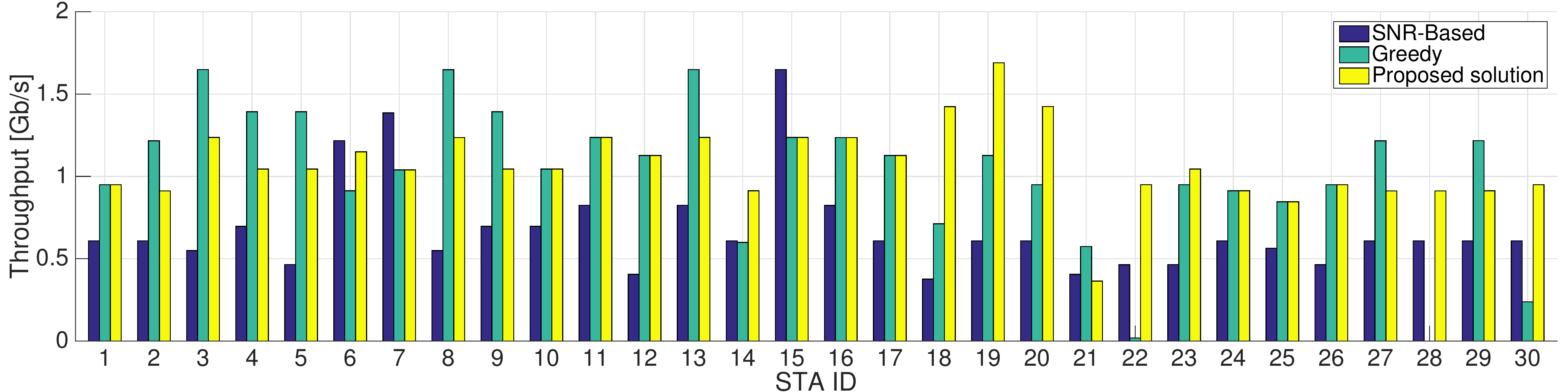}
	\caption{\rev{Individual throughputs achieved with the SNR-based, greedy, and proposed methods.}}
	\label{fig:SatPerf}
	\end{subfigure}
	\begin{subfigure}[t]{0.25\textwidth}
	\includegraphics[width=\textwidth]{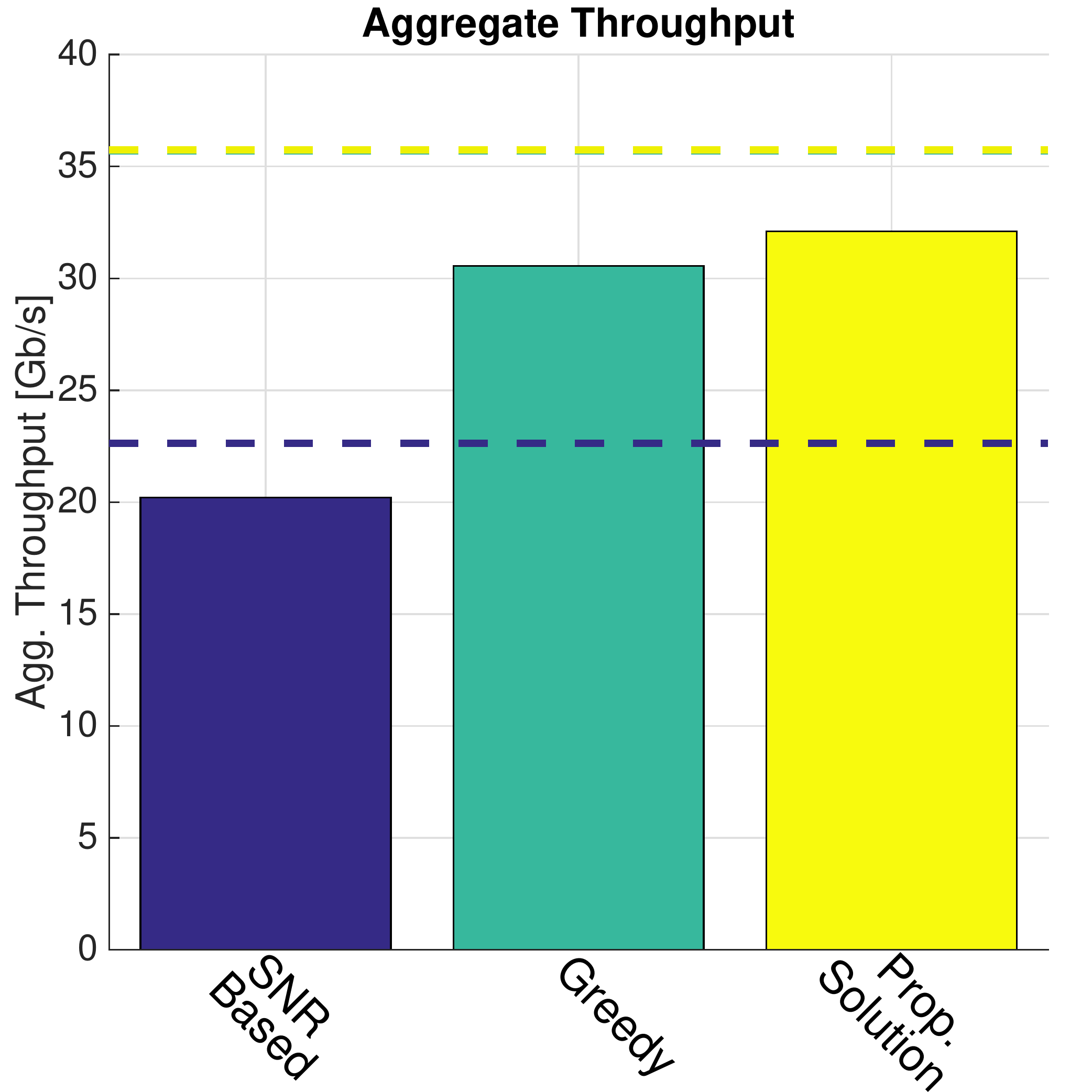}
	\caption{\rev{Aggregate throughput comparison.}}
	\label{fig:agg_tp_sat}
	\end{subfigure}

\caption{Enterprise mmWave network with 9 APs positioned on a grid and 30 client
	stations deployed randomly. Client associations enforced by and
	throughput performance attained with the SNR-based policy, \rev{the greedy association approach,} and
	respectively the proposed association control mechanism. All clients are
	backlogged (saturation conditions) and equal airtime allocation is
	performed at each AP. Theoretical maximum shown with dashed lines. Simulation results.}

\label{fig:Sat}
% \vspace*{-1.25em}
\end{adjustwidth}
\end{figure*}

\subsection{Backlogged Traffic}

We first investigate the performance of the association control scheme we
propose for saturation scenarios, which involves solving the relaxed utility
maximisation problem and executing an iterative rounding algorithm
(Sec.~\ref{sec:satcond}). For the scenario with 4 APs, we compare the
behaviour of our approach against the standard's default SNR-based association
control policy and against the optimal solution of problem (\ref{eq:goal}) obtained through exhaustive search.
For the scenario with 9 APs, we compare the throughput performance of our approach only against the standard's default policy \rev{and the greedy association previously described}.
For a fair comparison, in all cases we allocate equal airtime to all clients of each AP (proportional fairness), with any association schemes.

We first illustrate in Fig.~\ref{fig:SatPerf_4ap10sta} the individual station
throughputs achieved in the first scenario, where observe that with our scheme
all the clients attain superior (up to 74\% higher) throughput, as compared to
the SNR-based policy. Overall, the proposed solution achieves a 1\% utility
gain, which corresponds to a 53\% increase in aggregate network throughput,
which we show in Fig.~\ref{fig:agg_tp_sat_4ap10sta}. In addition, we run
exhaustive searches over the solution space, to quantify the difference between
the solution found by solving the relaxed problem and running the proposed
iterative rounding algorithm, and the absolute optimum (yellow bars in
Fig.~\ref{fig:Sat_4ap10sta}). Observe that this difference is negligible -- our
method attains small throughput gains (up to 1.7\%) for a subset of stations,
and slightly lower (up to 2.9\%) values for others. These small
differences are mainly due to practical channel conditions that lead to frame
collisions or reception failure, which are overlooked by the theoretical
problem formulation. Overall, the utility loss of our method is only 0.0002\%,
which corresponds to an aggregate network throughput loss of 0.0035\%. To add further perspective, in Fig.~\ref{fig:agg_tp_sat_4ap10sta} we also plot with dashed lines the theoretical maximum throughput obtained numerically as a function of the optimal solution returned by each approach considered, noting only a ~10\% difference in all cases.

Turning attention to the 9 APs scenario, to gain a deeper understanding of the individual and aggregate throughput performance, we also illustrate the Client--AP links established using the proposed, SNR-based, \rev{and greedy} methods. Note that using the SNR-based approach, client stations are clustered around the nearest AP, irrespective of their local density, as depicted in
Fig.~\ref{fig:SNRassoc}. In effect, they can employ superior PHY bit rates, but
often share a single AP with many others (e.g. 9 clients connected to the AP in
the centre of the grid), while a subset of APs remain unutilised (4 APs in this
case). \rev{The greedy approach fairs better as it distributes the load among access points, yet this is performed na\"ively, which leads to half (or less) the throughput performance of the SNR-based strategy for some clients (e.g. clients 22 and 30).} In contrast, our approach distributes clients among all APs with the goal of maximising network utility (sum of log throughputs), as shown in
Fig.~\ref{fig:SatAssoc}. As such, clients may transmit at lower PHY rates, but
are allocated more airtime, which translates into higher throughput.

Indeed, Fig.~\ref{fig:SatPerf} demonstrates that with our scheme the majority of
clients attain superior throughput performance (even $>$100\% higher), while
only a small fraction experience a minor performance hit, as compared to the
SNR-based policy. Overall, our proposal attains a 2.5\% utility gain, which
corresponds to a 60\% gain in the aggregate network throughput, as illustrated
in Fig.~\ref{fig:agg_tp_sat}.

We conclude that \textbf{our scheme achieves substantial performance
improvements under backlogged traffic conditions}. In what follows, we
investigate the performance of the mechanism we introduced in
Sec.~\ref{sec:finiteload} for finite load conditions.

\subsection{Finite Load Conditions}

Next we extend the performance analysis to finite load conditions, i.e. when stations have limited
traffic demands (offered load). We consider the same indoor topologies,
but with heterogeneous offered loads, whereby demand varies between
460Mb/s and 2.3Gb/s in the 4 APs scenario, and between 500Mb/s and
1.25Gb/s in the 9 APs scenario. Recall that maximising network utility in such circumstances,
requires not only to find the appropriate association matrix, but also the
airtimes allocated to each station at each AP (see Sec.~\ref{sec:finiteload}). 

For comparison, we analyse the performance of the proposed simulated annealing
and water filling based solution (``Proposed SA-WF solution'') and
that of:

\begin{itemize}
 \item SNR-based association and equal airtime (EA) allocation;
 \item SNR-based association with airtime water filling (WF); %based airtime allocation;
 \item The distributed DAA algorithm proposed in \cite{Athanasiou:2015};\footnote{Athanasiou \emph{et al.} only address the association problem and do not consider airtime allocation \cite{Athanasiou:2015}. As such, we use their approach with the same equal airtime (proportional fair) allocation strategy.}
 \item The optimal solution obtained through exhaustive search (only for the 4 APs topology).
\end{itemize}

\noindent Our simulated annealing algorithm works with the following parameters: $T0 = 20,
\alpha = 0.7, q = NM/2, Tmin = 0.001$, $p=0.1$, which we empirically found to
yield good performance, as we will show in Sec.~\ref{sec:runtime}.

\begin{figure*}[t]
%\vspace*{-0.5em}
\centering
	\begin{subfigure}[t]{0.74\textwidth}
	\centering
	\includegraphics[width=\textwidth]{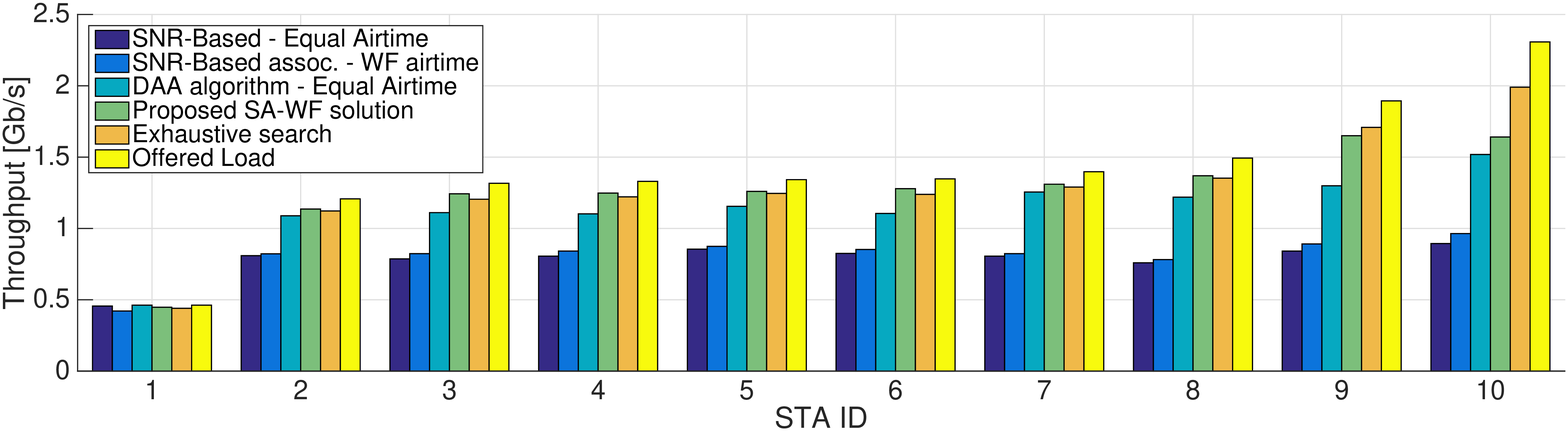}
	\caption{Offered loads and throughputs with SNR-based approaches, DAA, and the proposed solution.}
	\label{fig:LoadIndividual_4ap10sta}
	\end{subfigure}
	\begin{subfigure}[t]{0.25\textwidth}
	\includegraphics[width=\columnwidth]{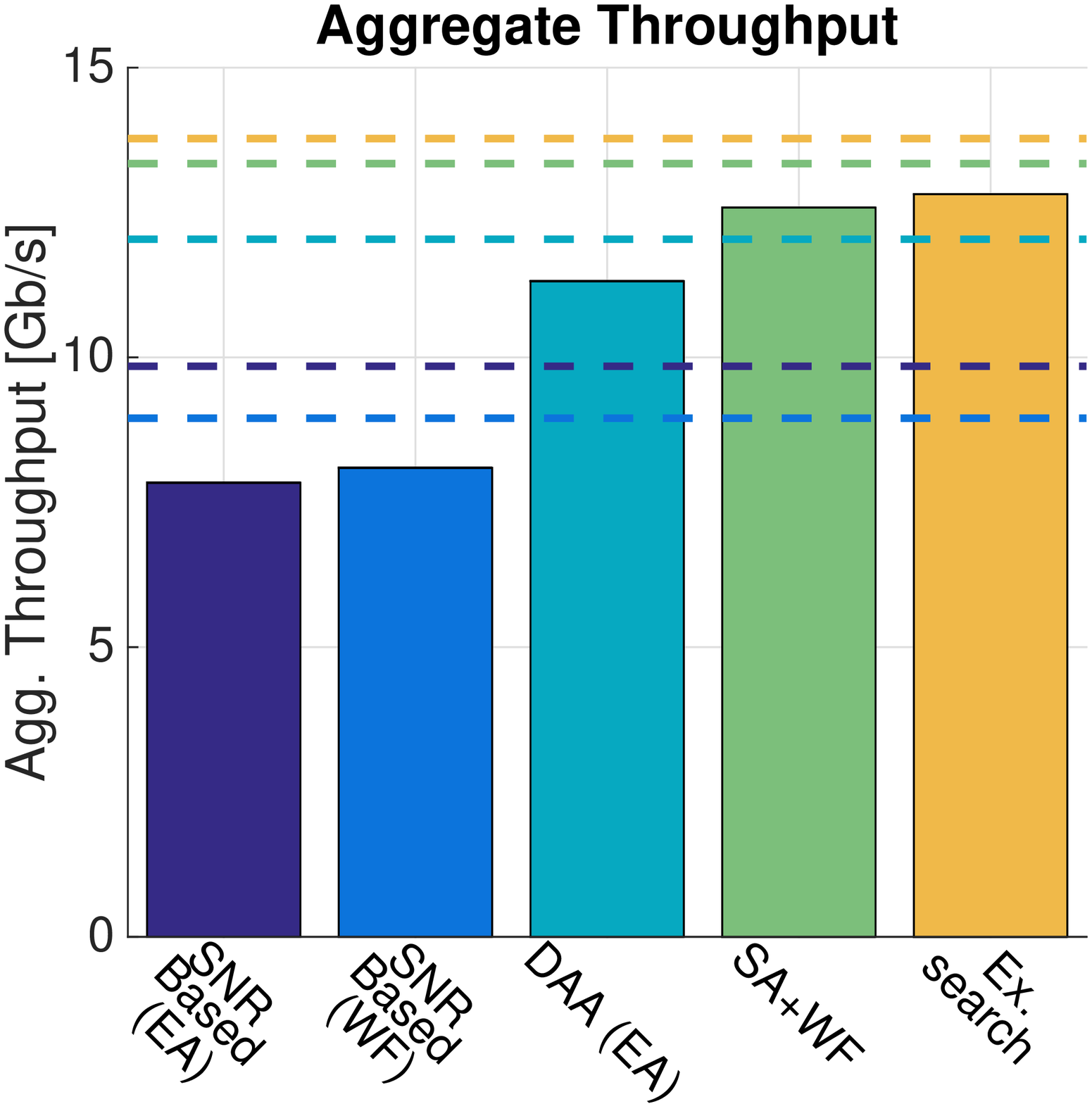}
	\caption{Aggregate throughput comparison.}% in the network deployment shown in Fig.~\ref{fig:LoadDAA}, when association is enforced by SNR-based policies, DAA \cite{Athanasiou:2015}, and our proposal (SA-WF). Simulation Results.}
	\label{fig:agg_tp_load_4ap10sta}
	\end{subfigure}

	\caption{Enterprise mmWave network with 4 APs positioned on a grid
	and 10 client stations deployed using the pmf reported in
	Fig.~\ref{fig:pmf4ap10sta}.  Client throughput performance attained with
	DAA \cite{Athanasiou:2015}, SNR-based policies, the proposed simulated
	annealing and water filling (SA-WF) based association control solution,
	and respectively exhaustive search. Clients have heterogeneous offered
	loads between 0.4--2.3Gb/s (finite load).  Theoretical maximum shown with dashed lines. Simulation results.}

\label{fig:LoadPerf_4ap10sta}
% \vspace*{-1.25em}
\end{figure*}

\begin{figure*}[!ht]
\begin{adjustwidth}{-1cm}{-1cm}
%\vspace*{-0.5em}
\centering
	\begin{subfigure}[t]{0.45\textwidth}
	\centering
	\includegraphics[width=\textwidth]{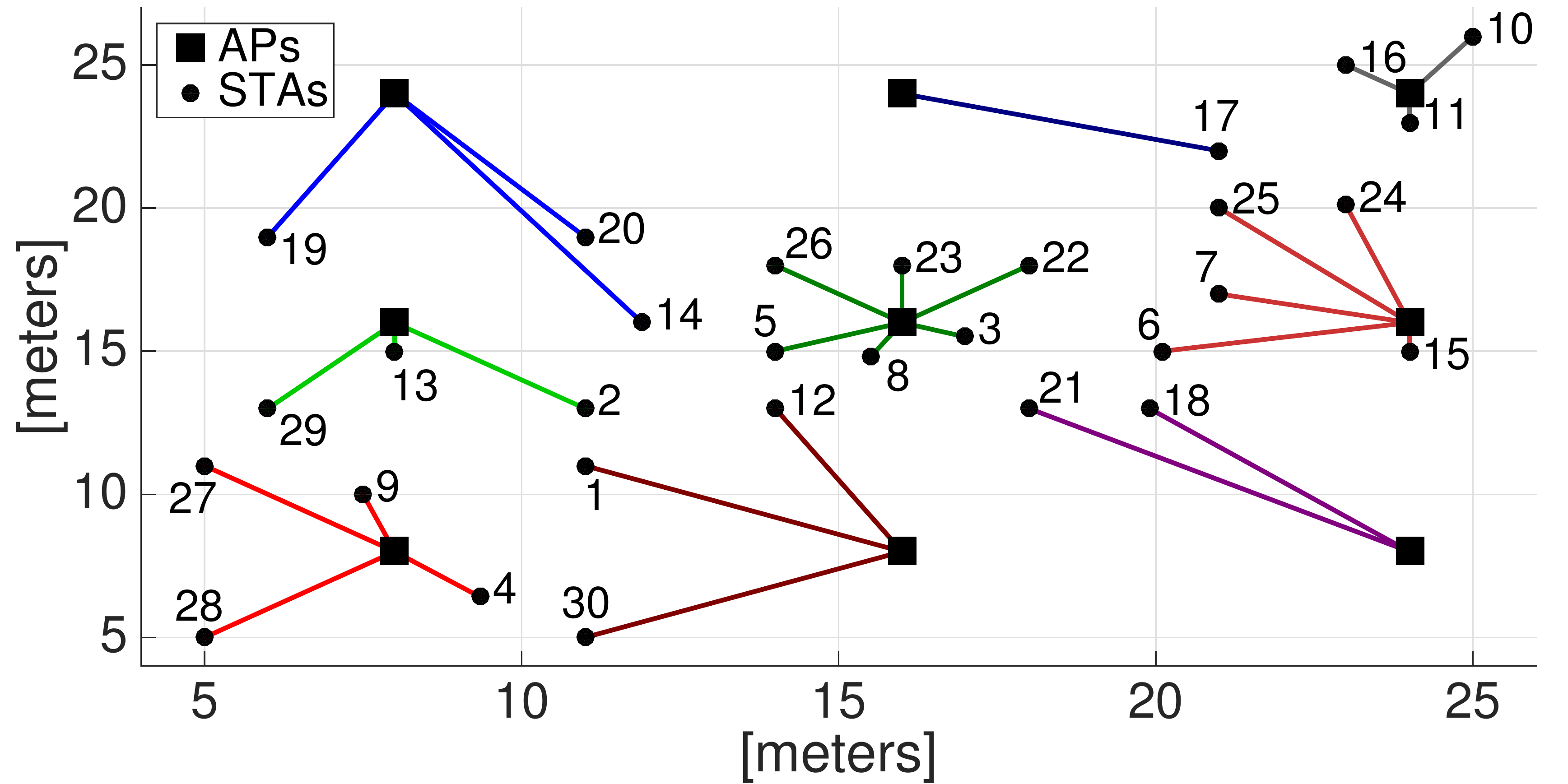}
	\caption{Client--AP links enforced by the DAA algorithm \cite{Athanasiou:2015}.}
	\label{fig:LoadDAA}
	\end{subfigure}
	\hspace*{1em}
	\begin{subfigure}[t]{0.45\textwidth}
	\centering
	\includegraphics[width=\textwidth]{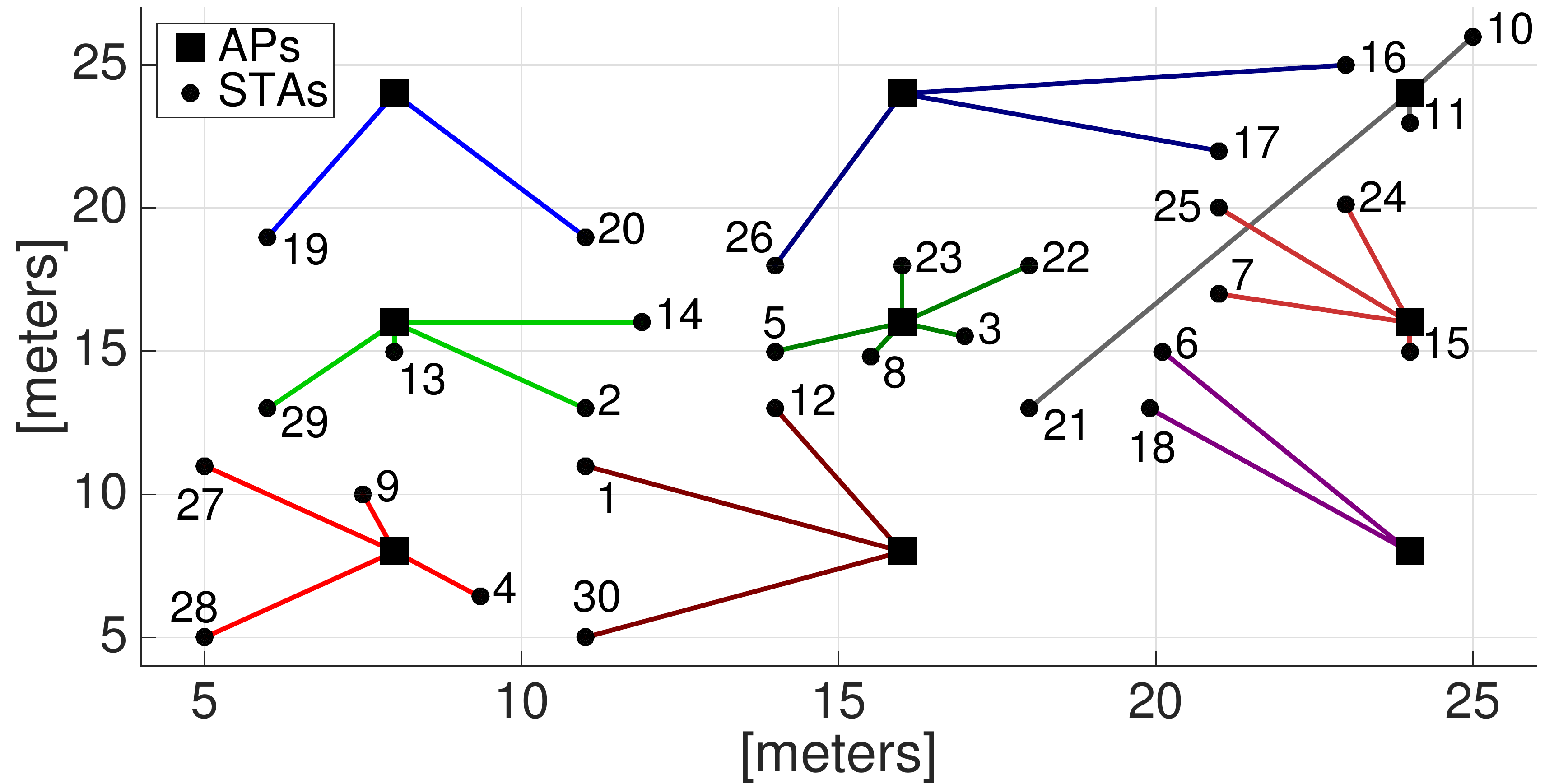}
	\caption{Client--AP links enforced by the proposed SA-WF scheme.}
	\label{fig:LoadAssoc}
	\end{subfigure}
\\	
	\begin{subfigure}[t]{.87\textwidth}
	\centering
	\includegraphics[width=\textwidth]{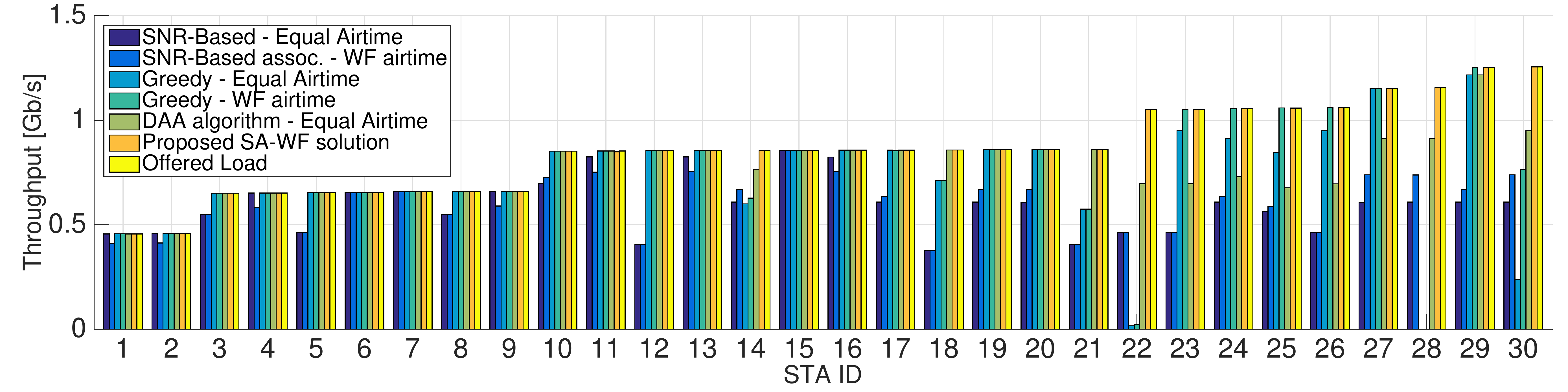}
	\caption{\rev{Offered loads and throughputs with SNR-based, greedy, and DAA approaches, and respectively the proposed solution.}}
	\label{fig:LoadIndividual}
	\end{subfigure}
	\begin{subfigure}[t]{0.23\textwidth}
	\centering
	\includegraphics[width=\columnwidth]{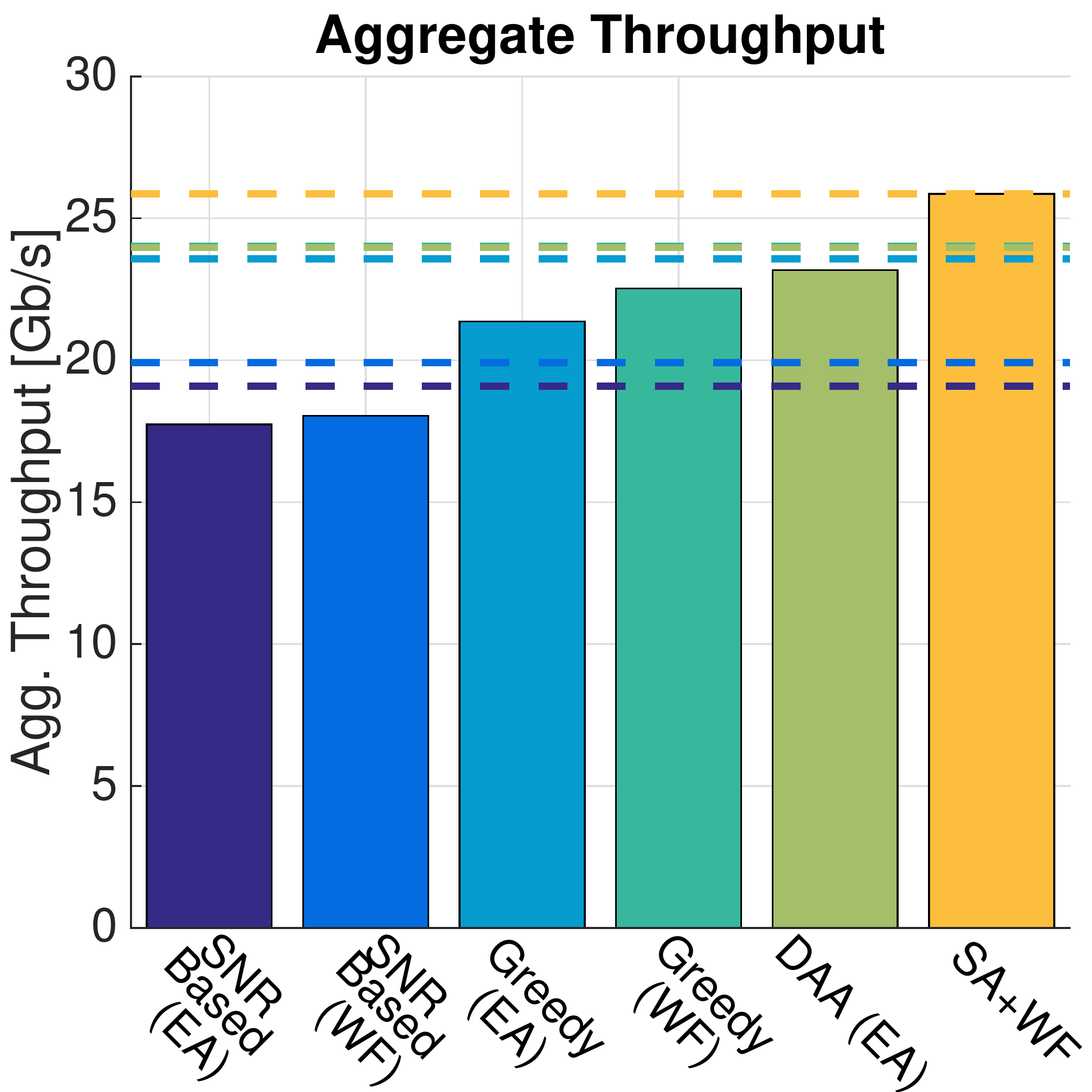}
	\caption{\rev{Aggregate throughput comparison.}}% in the network deployment shown in Fig.~\ref{fig:LoadDAA}, when association is enforced by SNR-based policies, DAA \cite{Athanasiou:2015}, and our proposal (SA-WF). Simulation Results.}
	\label{fig:agg_tp_load}
	\end{subfigure}

\caption{Enterprise mmWave network with 9 APs positioned on a grid and 30 client
	stations deployed randomly. Client associations enforced by and
	throughput performance attained with DAA \cite{Athanasiou:2015},
	SNR-based policies, and respectively the proposed simulated annealing
	and water filling (SA-WF) based association control solution. \rev{Also shown is the performance with the greedy association approach working with equal time allocation (EA) and with water-filling (WF).} Clients
	have heterogeneous offered loads between 0.5--1.25Gb/s (finite load). Theoretical maximum shown with dashed lines.
	Simulation results.}

\label{fig:LoadPerf}
% \vspace*{-1.25em}
\end{adjustwidth}
\end{figure*}

We illustrate the results of this experiments if Figs.~\ref{fig:LoadPerf_4ap10sta} and~\ref{fig:LoadPerf}, where we also show with yellow bars the offered load of each station. In the second case, we also depict the client associations enforced by our proposal (Fig.~\ref{fig:LoadAssoc}) and the DAA scheme (Fig.~\ref{fig:LoadDAA}). First, note that also under finite load conditions, our solution largely performs very close to the 
optimum obtained through exhaustive search (Fig.~\ref{fig:LoadIndividual_4ap10sta}), the difference in individual between the two being notable only at stations \#10 (21\%). Overall, our solution yield a 0.0004\% smaller network utility, which corresponds to an aggregate throughput loss of 1.8\% (hardly appreciable in Fig.~\ref{fig:agg_tp_load_4ap10sta}).
On the other hand, the aggregate offered load exceeds the resources available in the network, while DAA and the SNR-based policy perform worse that the proposed simulated annealing and water
filling based approach. In particular, observe in Fig.~\ref{fig:LoadIndividual_4ap10sta} that with our scheme almost all the clients
attain a superior throughput performance, namely up to 96\% and 27\% higher, as compared to the
SNR-based policy and DAA. Overall, we attain network utility gains up to 2.1\%,
corresponding to aggregate throughput gains of 11--60\%, as illustrated in
Fig.~\ref{fig:agg_tp_load_4ap10sta}.

Examining now the 9 AP topology, we note that DAA works distributively and thus manages to balance well the number of
clients across different APs, as seen in Fig.~\ref{fig:LoadDAA}. However, the
underlying assumption in this approach is that APs will always be able to
accommodate any traffic demand, which is impractical, while airtime allocation
at each AP is not considered explicitly. Consequently, although the network has
sufficient resources to accommodate all demands in this scenario, some stations only receive a
fraction of offered load with this approach (Fig.~\ref{fig:LoadIndividual}). In particular, as the
offered load increases, DAA largely accommodates only $\sim\sfrac{2}{3}$ of the
individual demands (stations 22--30).
The drawback of not explicitly accounting airtime at each AP is more obvious at the AP located in the bottom left corner. Even if the AP serves the same
subset of clients with both the DAA algorithm (Fig.~\ref{fig:LoadDAA})
and the proposed SA--WF scheme (Fig.~\ref{fig:LoadAssoc}), two of these clients (27 and 28)
experience superior performance with our proposal. This is because our
water filling procedure (Algorithm~\ref{alg:watfil}) takes into consideration
the actual offered loads when allocating airtime to each associated client; in contrast, simply allocating equal airtime to each
client with DAA associations proves sub-optimal. \rev{For completeness, we also compare the performance of our solution against the greedy approach, which maintains the same associations as shown in Fig.~\ref{fig:GreedyAssoc}. We consider this performs equal airtime allocation or is combined with the proposed water-filling scheme. Note that neither of the two match the performance of our solution or that of DAA. We remark that a greedy association of clients is only marginally better than an SNR-based approach under finite load conditions.}

\begin{figure*}[t]
%\vspace*{-0.5em}
\centering
	\begin{subfigure}[t]{0.25\textwidth}
	\centering
	\includegraphics[width=\columnwidth]{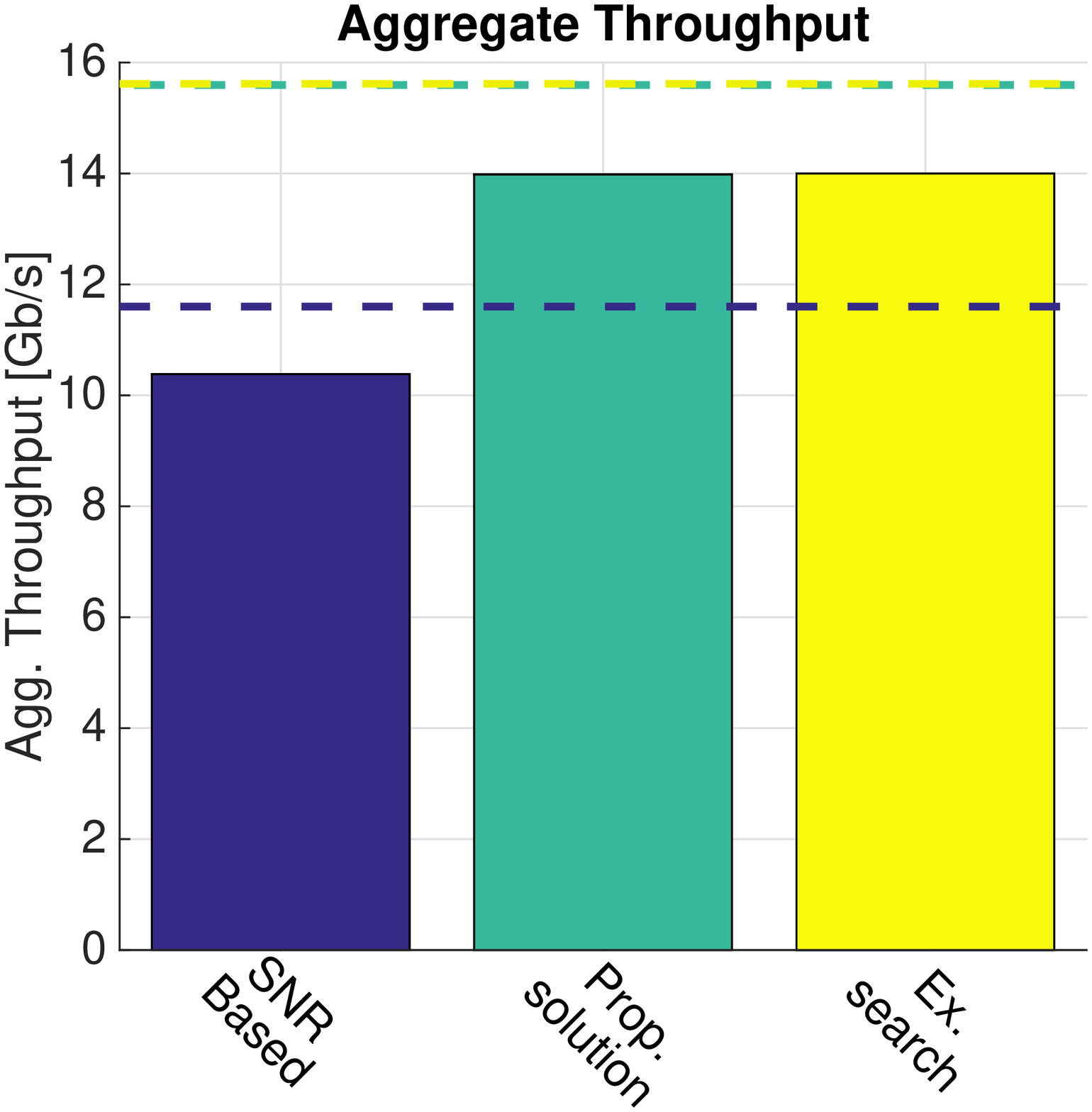}
	\caption{Aggregate throughput comparison.}
	\label{fig:agg_tp_sat_mm}
	\end{subfigure}
	%\hspace*{0.5em}
	\begin{subfigure}[t]{0.74\textwidth}
	\includegraphics[width=\textwidth]{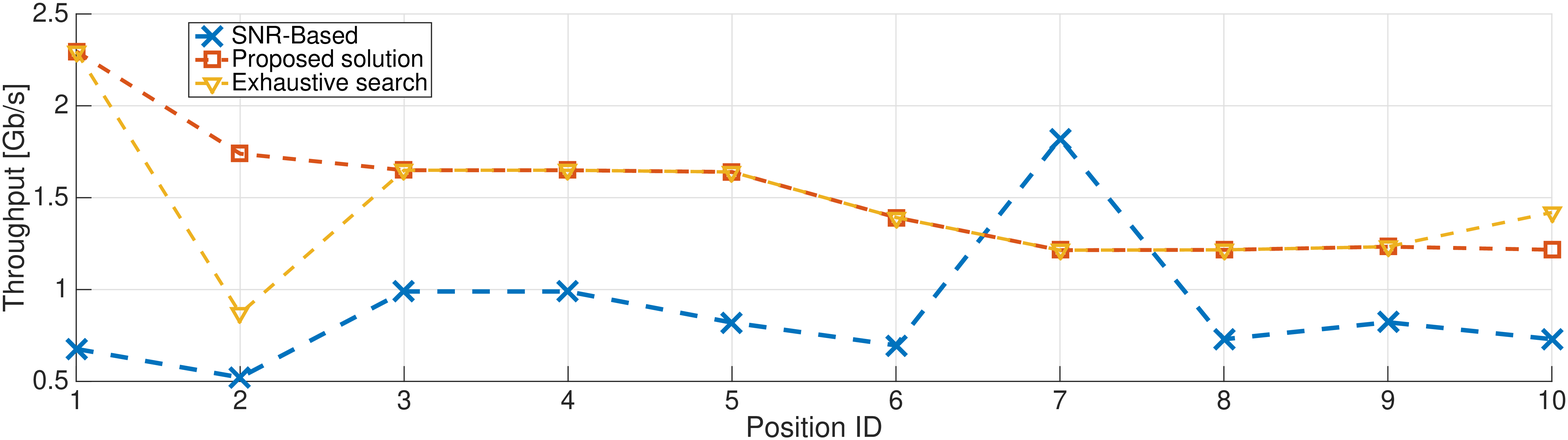}
	\caption{Throughput of a single mobile station at 10 different locations.}
	\label{fig:SatIndividual_mm}
	\end{subfigure}

	\caption{Enterprise mmWave network with 4 APs positioned on a grid and
	10 client stations moving following a random waypoint mobility model.
	Client throughput performance attained with the	SNR-based association
	policy, the proposed solution, and the optimum obtained via exhaustive
	search. All clients are backlogged (saturation conditions) and equal
	airtime allocation is performed at each AP. Theoretical maximum shown with dashed lines. Simulation results.}

\label{fig:SatPerf_mm}
% \vspace*{-1.25em}
\end{figure*}

\begin{figure*}[t]
%\vspace*{-0.5em}
\centering
	\begin{subfigure}[t]{0.25\textwidth}
	\centering
	\includegraphics[width=\columnwidth]{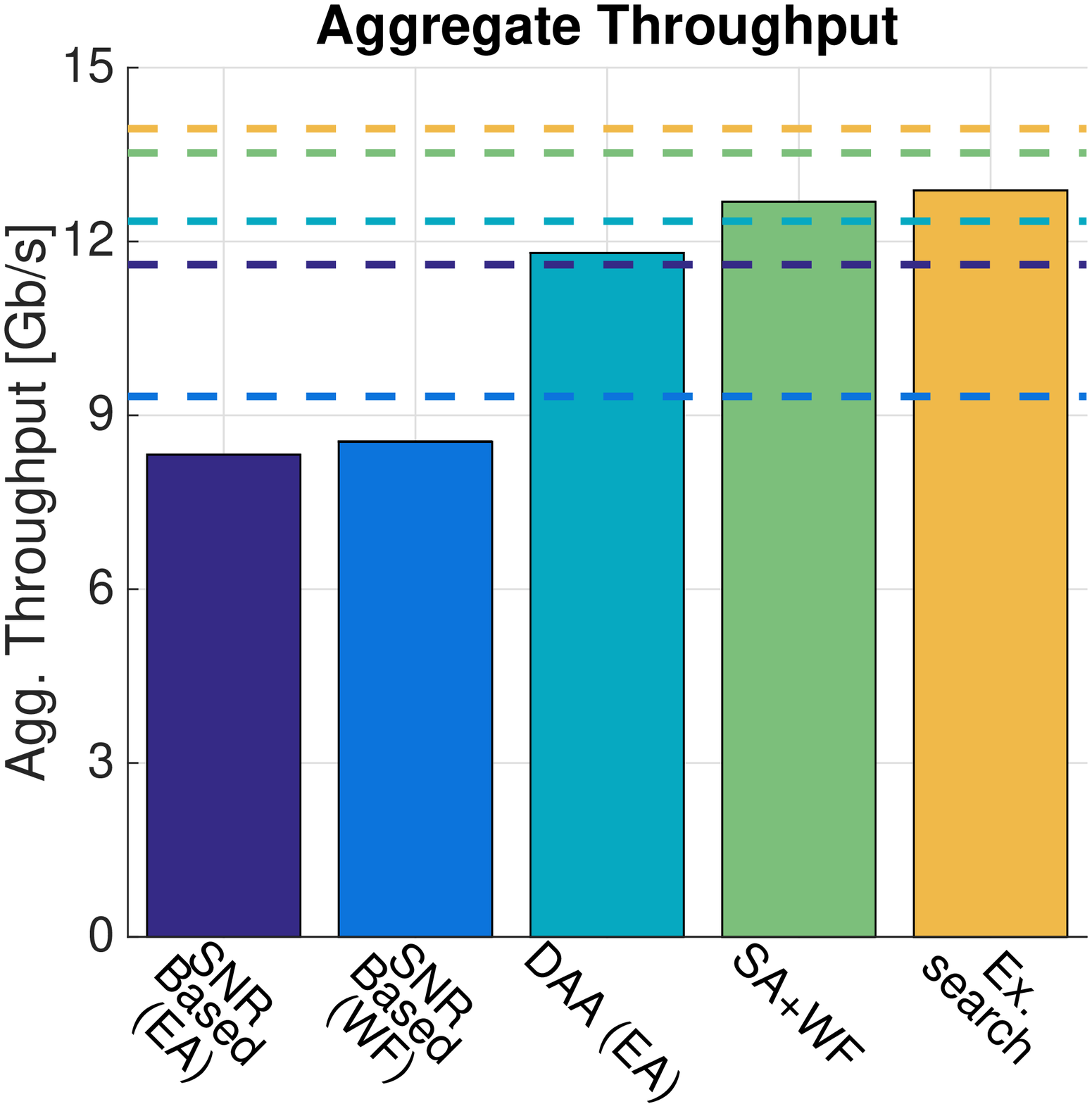}
	\caption{Aggregate throughput comparison.}
	\label{fig:agg_tp_load_mm}
	\end{subfigure}
	\begin{subfigure}[t]{0.74\textwidth}
	\includegraphics[width=\textwidth]{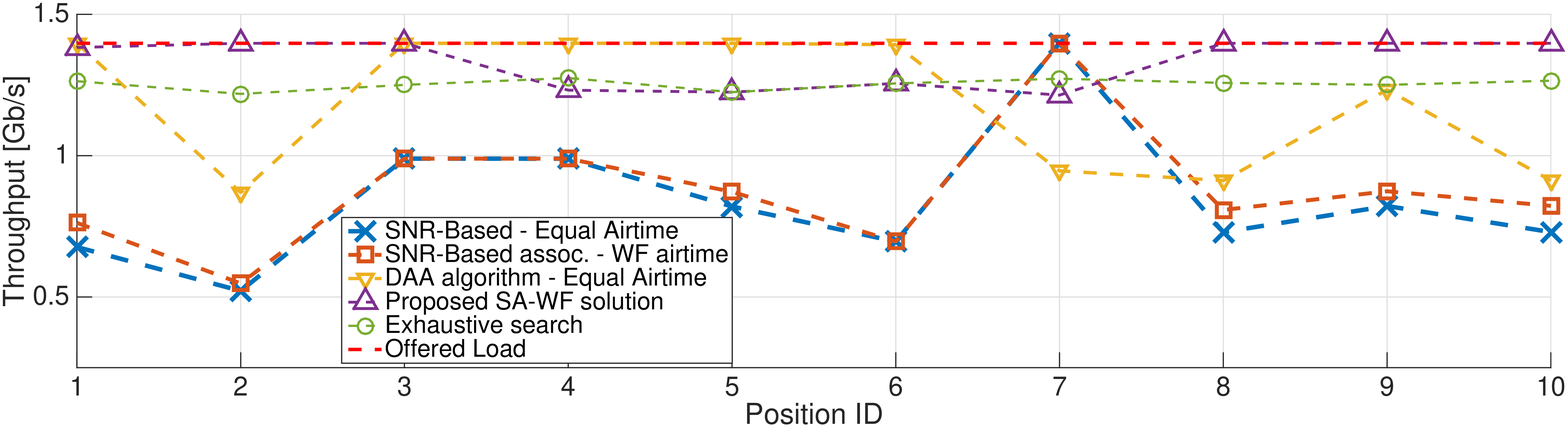}
	\caption{Offered loads and throughputs for a single station at 10 different locations.}
	\label{fig:LoadIndividual_mm}
	\end{subfigure}

	\caption{Enterprise mmWave network with 4 APs positioned on a grid
	and 10 client stations moving using a random waypoint mobility model. Client throughput performance attained with
	DAA \cite{Athanasiou:2015}, SNR-based policy, the proposed simulated
	annealing and water filling (SA-WF) based association control solution,
	and the optimum obtained via exhaustive search. Clients have heterogeneous offered
	loads between 0.4--2.3Gb/s (finite load). Theoretical maximum shown with dashed lines.  Simulation results.}

\label{fig:LoadPerf_mm}
% \vspace*{-1.25em}
\end{figure*}

The simulated annealing and water filling based solution we propose
obtains association and airtime allocation matrices that \textbf{successfully accommodate the demand of all stations (33\% more than
DAA)}. %, as shown in Fig.~\ref{fig:LoadIndividual}.
Overall, our proposal
increases network utility attainable with SNR-based policy and DAA by up to 2\%,
which effectively translates into an aggregate throughput gain of 11--32\%, as
shown in Fig.~\ref{fig:agg_tp_load}.

\subsection{User Mobility}
\label{sec:mobmod}

In the following we evaluate the performance of the proposed methods when users
move in the coverage area and either are backlogged or have finite demands.  To
this end, we work with a network deployment comprising 4 APs and 10 client
stations. The APs are positioned again as in Fig.~\ref{fig:pmf4ap10sta}, while
clients' positions evolve over time according to a random waypoint mobility
model. More specifically, we start by randomly deploying the stations in the
$[7, 23]$m$\times [7, 16]$m area, then simulate 100 seconds of user movements
with velocities randomly distributed between 0.2m/s and 2.2m/s, pause intervals
between movements uniformly distributed between 1s and 20s, and walk times
uniformly distributed between 1s and 5s.  We take ten snapshots of the client
positions obtained with this mobility model (one every 10 seconds) and measure
the user throughputs attained at each of these positions. We illustrate the
results of this experiments in Figs.~\ref{fig:SatPerf_mm} and
~\ref{fig:LoadPerf_mm}, where we plot the throughput of one station as a
function of the position, as well as the average of the total (aggregate)
throughput over all positions, in backlogged and respectively finite load
scenarios.

We observe that, similar to the static scenarios, when stations are backlogged
our proposal attains a 1.5\% network utility gain, which translates into a 35\%
aggregate throughput gain. Also in this case, our proposal is very close to the
optimal solution, as the difference is network utility is only 0.001\% and the
aggregate throughput only 0.0012\% lower with the proposed scheme. Further,
Fig.~\ref{fig:SatIndividual_mm} demonstrates the throughput attained by an
individual station with our approach is superior to that with the default
SNR-based policy in nine out of ten positions. It is important to note that,
although the SNR-based policy offers higher throughput for the sampled station
at position~\#7, this does not correspond to higher network-wide performance, as
confirmed by Fig.~\ref{fig:agg_tp_sat_mm}. In fact, our scheme offers 34.7\%
higher aggregate throughput, as compared to the SNR-based association policy.

We now examine a finite load condition scenario and compare the performance of
the proposed simulated annealing and water filling based approach, as well as
all other aforementioned schemes. Once again, our proposal attains superior
results as compared to SNR-based and DAA mechanisms, while its performance is
very close to the absolute optimal solution, as shown if
Fig.~\ref{fig:LoadPerf_mm}. In particular, the simulated annealing and water
filling mechanism attains a utility gain up to 2.8\% higher than that of
SNR-based and DAA approaches, which translates into an aggregate throughput gain
of 38--52\% (Fig.~\ref{fig:agg_tp_load_mm}). As compared to the optimal solution
obtained through exhaustive search the network utility loss is limited to
0.003\%, corresponding to only 0.012\% lower aggregate throughput. Taking a
closer look at the throughput of a single station at different locations,
Fig.~\ref{fig:LoadIndividual_mm} confirms the SNR-based policy only offers
superior performance to that of the proposed SA-WF and existing DAA scheme in
one location (position~\#7). Our solution meets the offered load at 6/10
locations and the throughput is very close to that at the other 4/10, unlike DAA
which works well in 5/10 locations, but offers significantly lower throughputs
at the other 5/10. It is also interesting to note that the optimal solution
(dashed green line) obtained through exhaustive search always under-performs in
this example.
This is easily explained by the fact that, in this particular example, the optimal solution
reduces the performance of the sampled station, while improving the performance of other stations, in order to maximise the network utility.
This is confirmed by the results we report in Fig.~\ref{fig:agg_tp_load_mm}, where
we observe the solution found through exhaustive search obtains the highest aggregate
throughput.

\subsection{\rev{Impact of Obstacles}}

\rev{In this subsection we study the impact of obstacles on the association derived and throughput obtained with the proposed simulated annealing and water-filling based scheme, as well as with the SNR-based and DAA benchmarks. We consider a similar office environment which is now partitioned with walls. 9 APs are placed again on a grid and clients are randomly deployed, lying in different parts of the layout as shown in Fig.~\ref{fig:obstsnr}. In this setting, links between clients and different APs are subject to obstacles. We consider finite heterogeneous load conditions (0.5--1.25Gb/s) and report the behaviour of all approaches in Fig.~\ref{fig:obstacles}. Observe that the attenuation due to obstacles impacts on the association decision of all schemes, making association to the APs placed in the bottom-right part of the layout particularly problematic. This is indeed observable by comparing Figs.~\ref{fig:obstsnr}--\ref{fig:obstsiman} with Figs.~\ref{fig:SNRassoc}, \ref{fig:LoadDAA}, and~\ref{fig:LoadAssoc}. As a result the offered load of fewer clients can be satisfied, which results in overall lower aggregate throughput for all approaches. Nonetheless, the total throughput attained by the proposed solution exceeds that offered by the benchmarks considered, as seen in Fig.~\ref{fig:agg_tp_obst}.}

\begin{figure*}
\begin{adjustwidth}{-1cm}{-1cm}
%\hspace*{-1.5in}
\centering
	\begin{subfigure}[t]{0.35\textwidth}
	\centering
	\includegraphics[width=\textwidth]{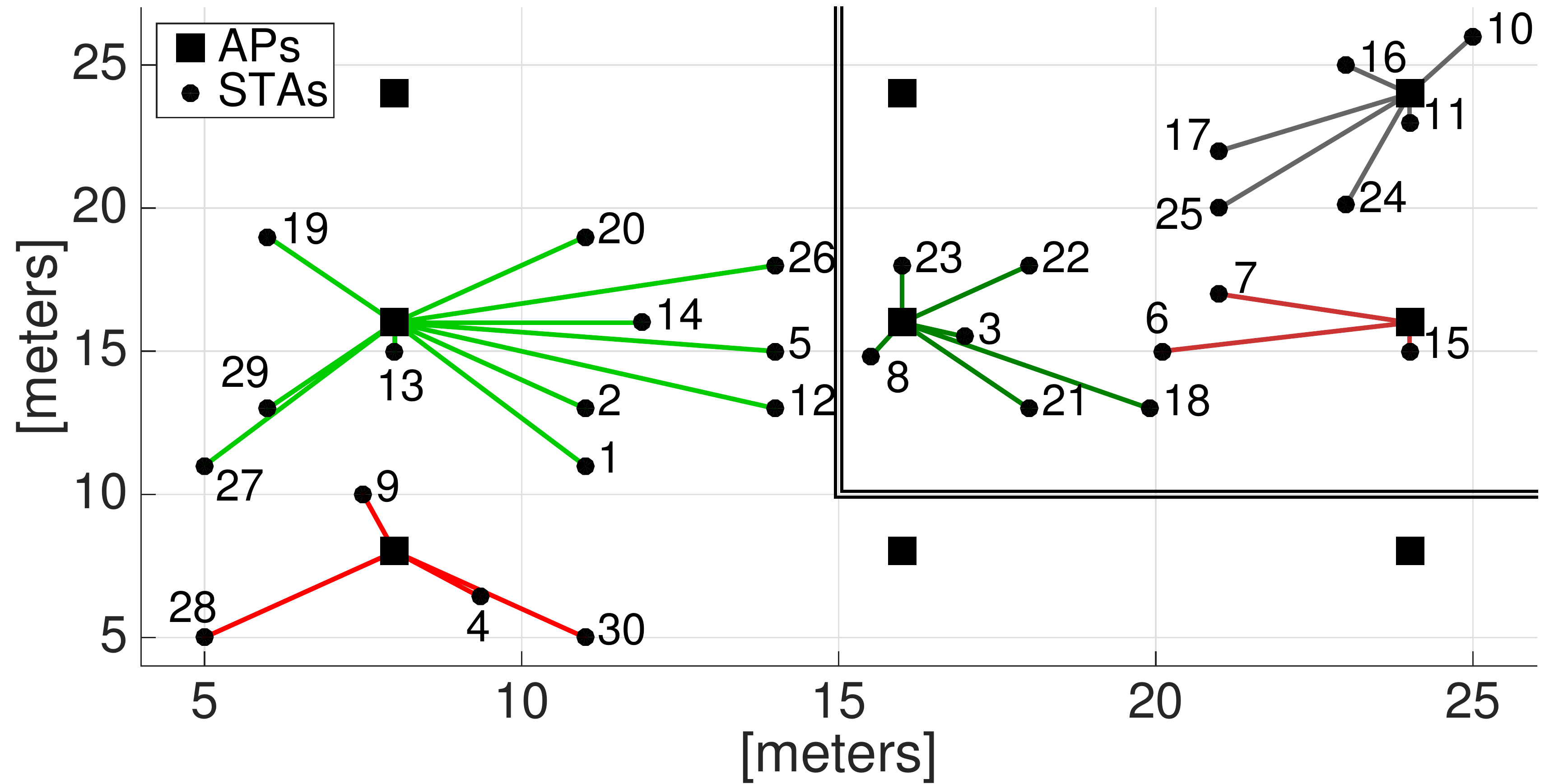}
	\caption{\rev{Client--AP links established using highest SNR policy.}}
	\label{fig:obstsnr}
	\end{subfigure}
	\hspace*{0.45em}
	\begin{subfigure}[t]{0.35\textwidth}
	\centering
	\includegraphics[width=\textwidth]{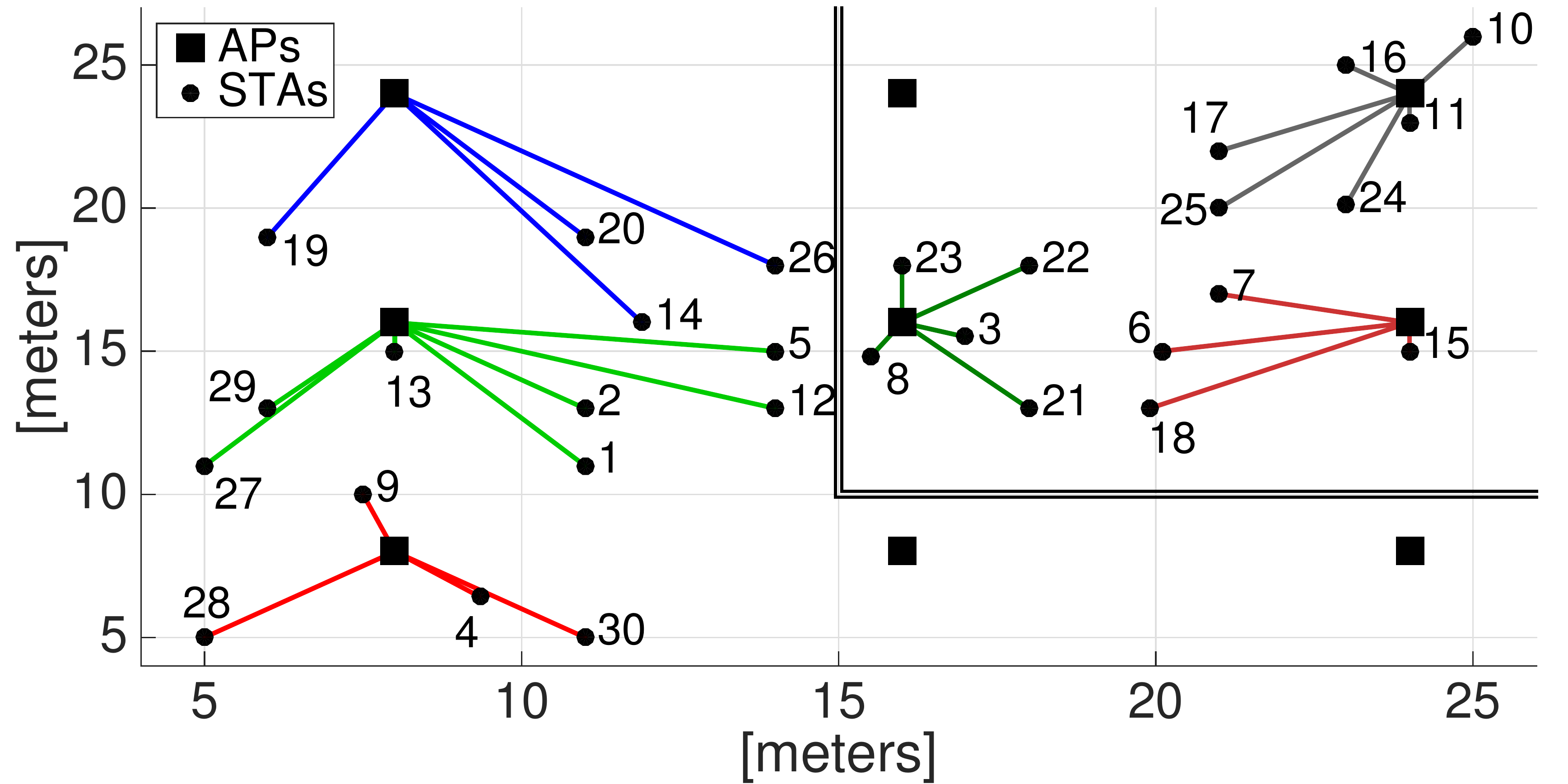}
	\caption{\rev{Client--AP links obtained with the DAA algorithm.}}
	\label{fig:obstdaa}
	\end{subfigure}	
	\hspace*{0.45em}
	\begin{subfigure}[t]{0.35\textwidth}
	\centering
	\includegraphics[width=\textwidth]{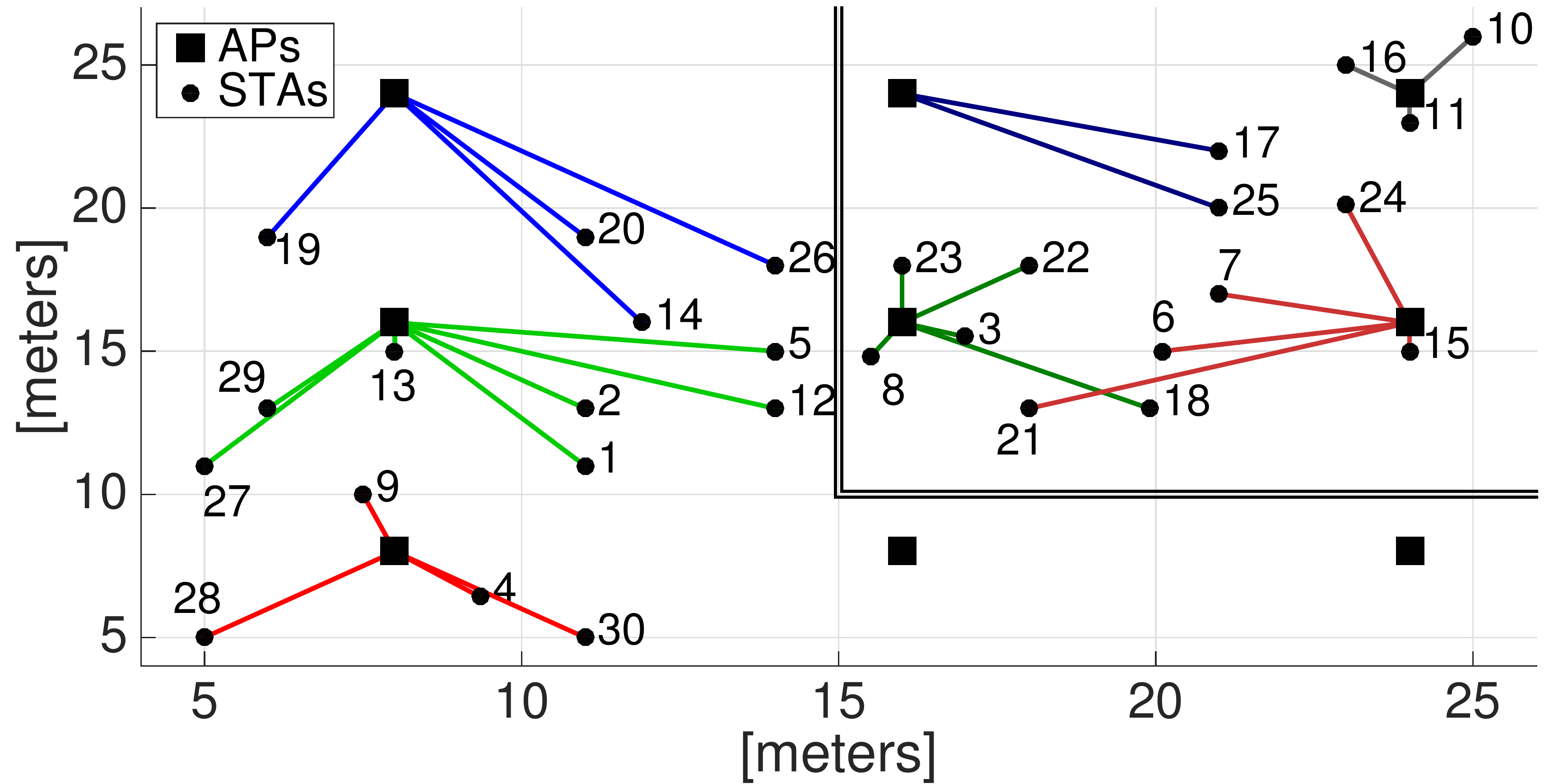}
	\caption{\rev{Client--AP links enforced by the proposed association algorithm.}}
	\label{fig:obstsiman}
	\end{subfigure}	
%\vspace*{1em}
	\begin{subfigure}[t]{.85\textwidth}
	\centering
	\includegraphics[width=\textwidth]{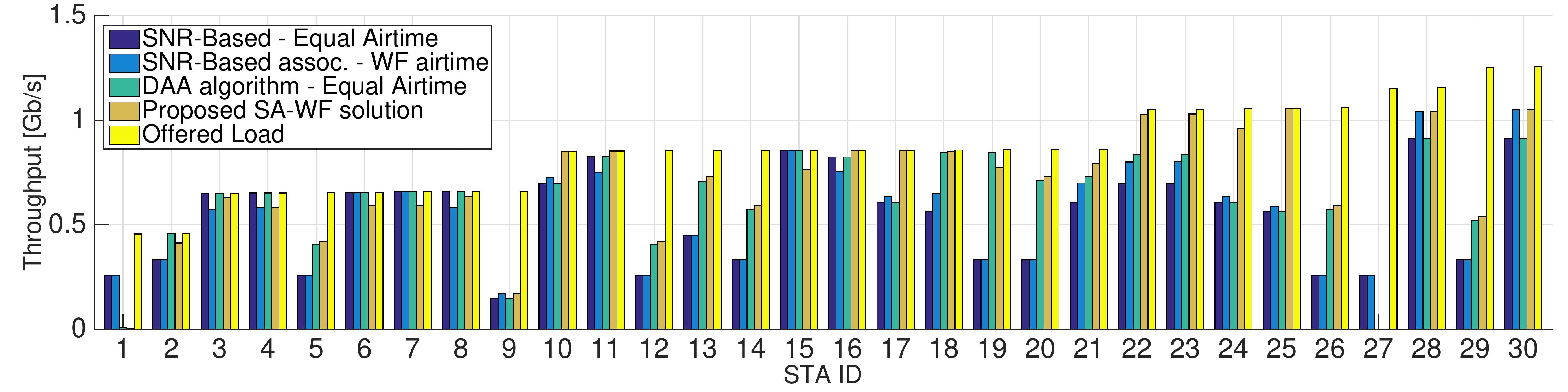}
	\caption{\rev{Individual throughputs achieved with the SNR-based, DAA, and proposed methods.}}
	\label{fig:obstperf}
	\end{subfigure}
	\begin{subfigure}[t]{0.25\textwidth}
	\includegraphics[width=\textwidth]{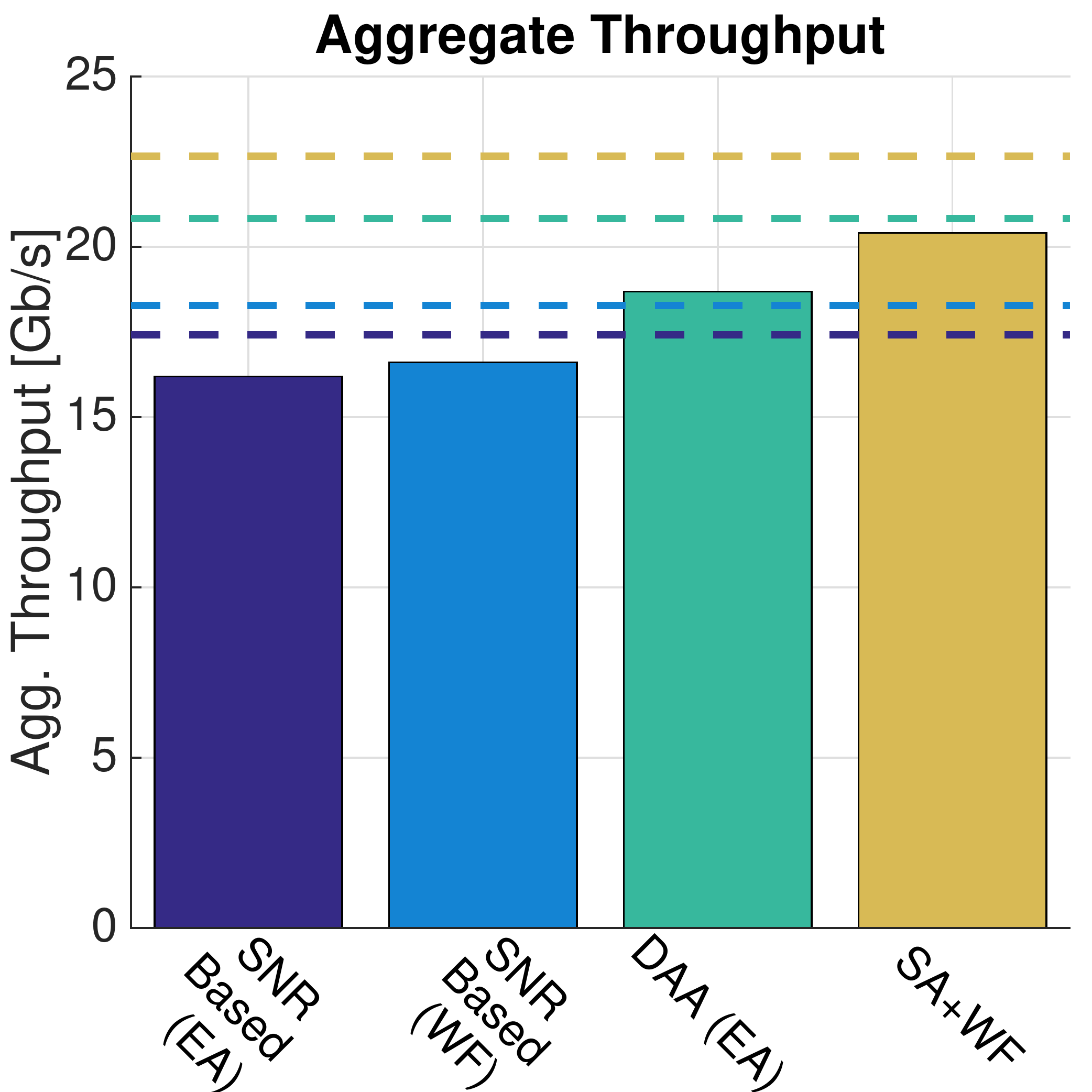}
	\caption{\rev{Aggregate throughput comparison.}}
	\label{fig:agg_tp_obst}
	\end{subfigure}

\caption{\rev{Enterprise mmWave network with 2 rooms separated by walls. 9 APs positioned on a grid and 30 client
	stations deployed randomly. Client associations enforced by and
	throughput performance attained with SNR-based, DAA, and respectively the proposed simulated annealing
	and water filling (SA-WF) based association control schemes. Clients
	have heterogeneous offered loads between 0.5--1.25Gb/s (finite load). Theoretical maximum shown with dashed lines.
	Simulation results.}}

\label{fig:obstacles}
% \vspace*{-1.25em}
\end{adjustwidth}
\end{figure*}

\subsection{Runtime Performance}
\label{sec:runtime}

Finally, we demonstrate that the proposed simulated annealing algorithm finds a
solution rapidly, making it suitable for real-time operation in an enterprise
mmWave network with a central controller. To this end we first take a closer look at
the algorithm's runtime in the more complex network deployment scenario considered previously (9APs and 30 clients) and examine the
utility at each step of the exploration.

As shown in Fig.~\ref{fig:iterations}, the algorithm starts with the solution of
the saturation problem, computes the utilities of each explored candidate
solution, and in this case accepts all of them,  even if the energy is negative
(line \ref{alg:siman:keepprob} of Algorithm~\ref{alg:siman}). By this procedure
\textbf{the algorithm finds an optimum that satisfies all clients' offered loads
within only 7 iterations}. This confirms that the simulated annealing parameters
$T0, \alpha, q, Tmin$, and $p$ we use are appropriate for the problem at hand.

\rev{To add further perspective, we simulate 400 topologies with 30--45 client stations and 9APs placed on a grid as before, and measure the execution time on a PC equipped with an Intel i7 CPU running at 3.1GHz. We plot the empirical CDF of these execution times in Fig.~\ref{fig:ecdf}. As expected, the runtime grows linearly with the number of clients, yet the median ranges between 180--355ms, which confirms the practical feasibility of our approach even in very dense settings where dynamics associated with pedestrian mobility can be tracked.}

\begin{figure*}[t]
%\begin{adjustwidth}{-2.5cm}{-2.5cm}
\centering

	\begin{subfigure}[t]{0.57\textwidth}
	\includegraphics[width=\columnwidth]{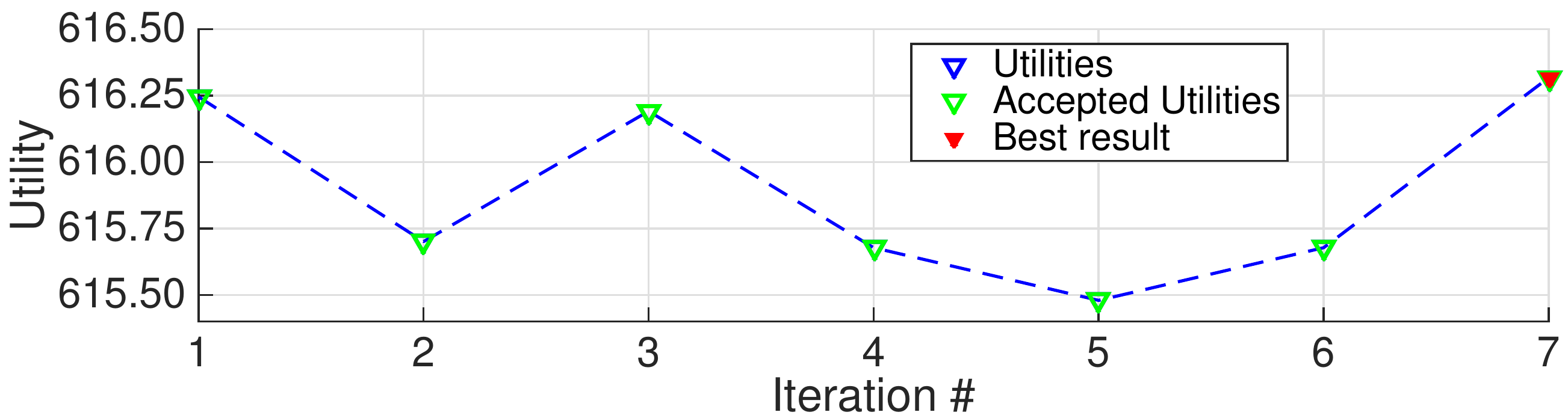}
	\caption{\rev{Number of iterations required to compute the solution for the topology shown in Fig.~\ref{fig:LoadAssoc}}}
	 \label{fig:iterations}
	\end{subfigure}
	\hspace*{0.5em}
	\begin{subfigure}[t]{0.4\textwidth}
	\includegraphics[width=\textwidth]{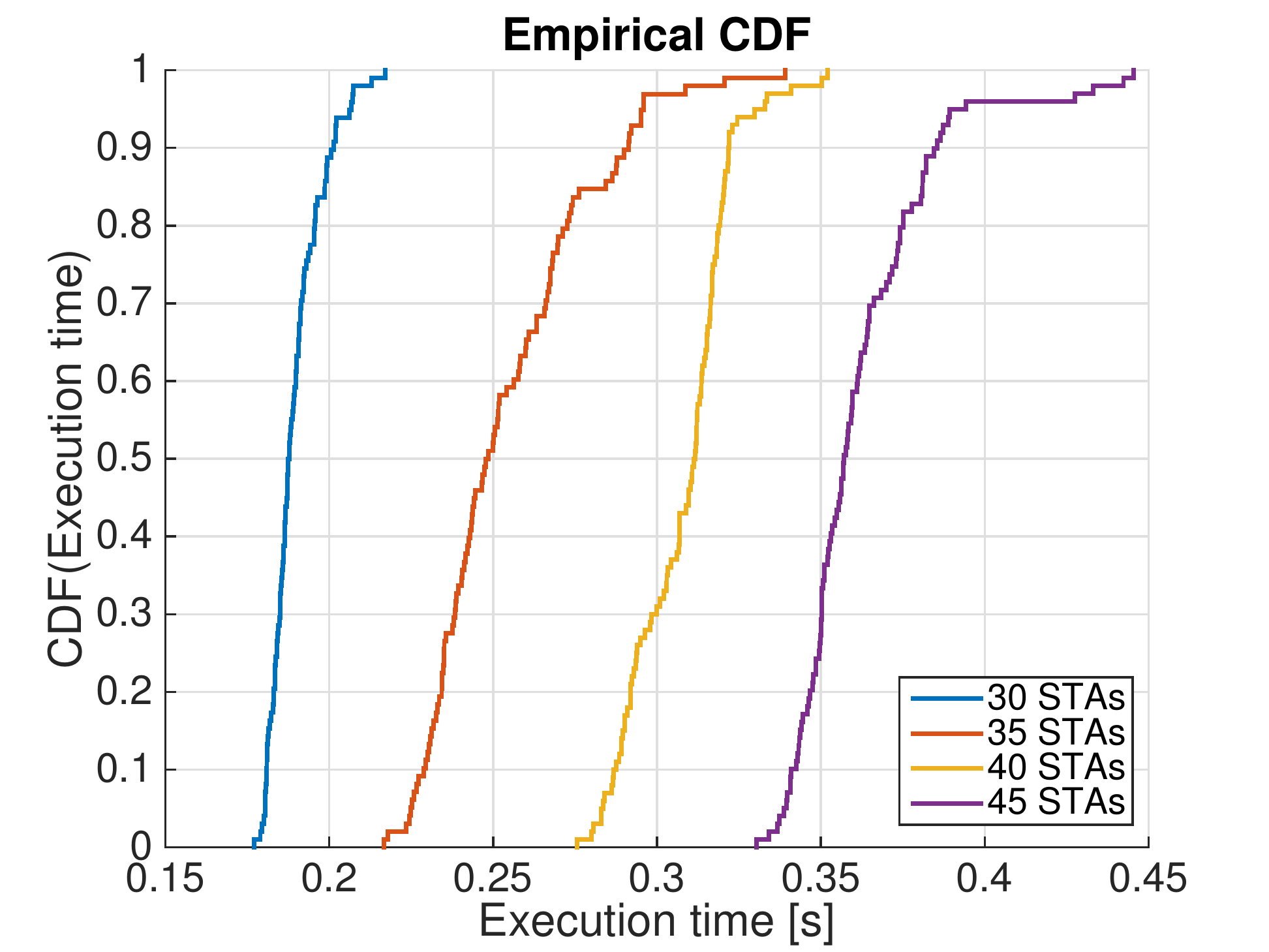}
	\caption{\rev{Empirical CDF of the execution time over 400 topologies with different sizes.}}
	 \label{fig:ecdf}
	\end{subfigure}
	
 \caption{\rev{Runtime of the proposed simulated annealing algorithm. Simulation results.}}
 \label{fig:runtime}
 %\vspace*{-0.5em}
% \end{adjustwidth}
\end{figure*}

%\section{Implementation Considerations}
%\label{sec:implementation}
%\input{implementation}

\section{Related Work}
\label{sec:related}
\noindent \textbf{60GHz Characterisation:} The feasibility of using mmWave technology for Gbps wireless connectivity has been the focus of several studies in recent years \cite{Singh:2011,Rappaport:2013,Zhu:2014,Sur:2015}. 
Singh \emph{et al.} contend that highly-directional links are feasible, but introduce terminal ``deafness'', shifting the focus from interference management to scheduling \cite{Singh:2011}. \rev{Interference regimes and the impact of mmWave base stations density are studied in~\cite{rebato2016understanding}.}
Rappaport \emph{et al.} employ prototype hardware to show steerable directional antennas work well over mmWave frequencies and have potential to support growing consumer data rates \cite{Rappaport:2013}, while extensive measurements further confirm the feasibility of 60GHz outdoor pico-cells~\cite{Zhu:2014}. %A re-configurable mmWave experimentation platform is yet to emerge. However, 
\rev{The feasibility of mmWave technology for top of the rack wireless communication in data centres was also demonstrated experimentally~\cite{flyways,Zhu:2014b}. Characterisation of indoor 802.11ad network performance, interference, and energy consumption has been undertaken in~\cite{Nitsche:2015,Saha2017}.} In addition, software-radio based studies in office environments reveal 802.11ad networks can achieve high throughput coverage beyond a single room \cite{Sur:2015}. These particularities are a key driver for the association control problem in enterprise mmWave, which we address herein.
\vspace*{0.5em}

\noindent \textbf{Association Control:} User association was studied widely in the context of both Wi-Fi \cite{Gong:2012,Athanasiou:2009,Xu:2010,Eretin:2008} and cellular systems \cite{Son:2009,Corroy:2012,Ye:2013}. In multi-cell 802.11 networks, game theoretic approaches were employed to balance load \cite{Xu:2010,Eretin:2008}, while heuristics were proposed for scenarios where legacy clients share the network with high-throughput (802.11n) stations \cite{Gong:2012}. Optimal association in wireless mesh networks is tackled using an airtime-metric based mechanism in~\cite{Athanasiou:2009}.

Son \emph{et al.} address association control in cellular networks by jointly optimising partial frequency reuse and load-balancing schemes~\cite{Son:2009}. For multi-tier cellular deployments, a theoretical cell association framework is introduced in \cite{Corroy:2012} and upper bounds on the achievable sum rate are derived. Ye \emph{et al.} combine utility maximisation and simple biasing approaches, to achieve user association with load balancing goals \cite{Ye:2013}. 

Research into client association in 60GHz networks is sparse. Athansiou \emph{et al.} address this problem from a load balancing perspective \cite{Athanasiou:2015}, though assume APs can always accommodate the demand of all clients, which is impractical, and pursue minimisation of maximum AP utilisation. They do not address distribution of APs' resource among clients.

Unlike previous work, we attack client association in enterprise mmWave networks as a utility maximisation problem under both backlogged and finite load scenarios, and heterogeneous link rates. We give low complexity algorithms that achieve close to optimal performance, while ensuring fair airtime allocation at each AP. Our schemes 
are amendable to deployment on emerging SDN enabled infrastructure supporting 802.11v/k management amendments.

\section{Conclusions}
\label{sec:conclusion}
In this paper we tackled network utility maximisation in high-end mmWave networks, capturing distinctive terminal deafness and user demand constraints, as well as dissimilar link qualities. Despite inherent NP-completeness, for backlogged conditions we solved a relaxed version of the problem and gave a low-complexity rounding algorithm that attains near-optimal performance. For finite load scenarios, we proposed a mechanism that combines simulated annealing and water filling techniques to find both the optimal association matrix and airtime allocation vector. Using an \mbox{NS-3} simulation tool we developed, we undertook a comprehensive evaluation campaign and showed that our solutions attain 60\% higher throughput as compared to the standard's default SNR-based policy, whilst accommodating the demand of 33\% more clients, as compared to recently proposed distributed association algorithms for mmWave networks.

%As mmWave hardware suitable for large scale experimentation will become available, we plan to validate our schemes in a real-life test bed, experimenting with bandwidth intensive and low latency applications, while shedding light on the implications of beam training and message control overhead on network performance. At this stage, we give an overview of key practical aspects to be considered for deploying our solution on SDN-enabled infrastructure that supports the latest 802.11 management amendments.

\section*{Acknowledgements}
The work of F. Gringoli was partly supported by the European Commission (EC) in the framework of the H2020-ICT-2014-1 project WiSHFUL (Grant agreement no. 645274).

%\vspace*{-0.5em}
%\section*{References}
{
\small
\bibliographystyle{elsarticle-num}

}

\end{document}